\documentclass[journal]{IEEEtran}

\ifCLASSINFOpdf
\else
\fi

\makeatletter
\def\ps@headings{
\def\@oddhead{\mbox{}\scriptsize\rightmark \hfil \thepage}
\def\@evenhead{\scriptsize\thepage \hfil \leftmark\mbox{}}
\def\@oddfoot{}
\def\@evenfoot{}}

\usepackage{amsmath,amssymb}
\usepackage{amsthm}
\usepackage{color}
\usepackage[mathscr]{eucal}
\usepackage{graphics,graphicx,multicol}
\usepackage{epsfig}
\usepackage{enumerate}
\usepackage{subfigure}
\usepackage{algorithmic}
\usepackage[section]{algorithm}
\usepackage{morefloats}
\usepackage{enumerate}
\usepackage{subfigure}
\usepackage{multirow}
\usepackage{array}
\usepackage{pstricks, pst-node, pst-plot, pst-circ}
\usepackage{moredefs}
\usepackage{multirow}
\usepackage{cite}
\usepackage{courier}
\newtheorem{thm}{Theorem}[section]
\newtheorem{lem}[thm]{Lemma}
\newtheorem{cor}[thm]{Corollary}
\newtheorem{rem}[thm]{Remark}

\begin{document}

\title{Optimal Radio Resource Allocation for Hybrid Traffic in Cellular Networks: Centralized and Distributed Architecture}
\author{Mo~Ghorbanzadeh,~\IEEEmembership{}
        Ahmed~Abdelhadi,~\IEEEmembership{}
        and~Charles~Clancy~\IEEEmembership{}
\thanks{Part of this work was accepted at IEEE ICNC CNC Workshop 2015 \cite{GhorbanzadehCNC2015_1}.

M. Ghorbanzadeh, A. Abdelhadi, and C. Clancy are with the Hume Center for National Security and Technology, Virginia Tech, Arlington,
VA, 22203 USA e-mail: \{mgh, aabdelhadi,tcc\}@vt.edu.}}

%\markboth{Journal of \LaTeX\ Class Files,~Vol.~11, No.~4, December~2014}%
%{Shell \MakeLowercase{\textit{et al.}}: Bare Demo of IEEEtran.cls for Journals}

\maketitle

\begin{abstract}
Optimal resource allocation is of paramount importance in utilizing the scarce radio spectrum efficiently and provisioning quality of service for miscellaneous user applications, generating hybrid data traffic streams in present-day wireless communications systems. A dynamism of the hybrid traffic stemmed from concurrently running mobile applications with temporally varying usage percentages in addition to subscriber priorities impelled from network providers' perspective necessitate resource allocation schemes assigning the spectrum to the applications accordingly and optimally. This manuscript concocts novel centralized and distributed radio resource allocation optimization problems for hybrid traffic-conveying cellular networks communicating users with simultaneously running multiple delay-tolerant and real-time applications modelled as logarithmic and sigmoidal utility functions, volatile application percent usages, and diverse subscriptions. Casting under a utility proportional fairness entail no lost 
calls for the proposed modi operandi, for which we substantiate the convexity, devise computationally efficient algorithms catering optimal rates to the applications, and prove a mutual mathematical equivalence. Ultimately, the algorithms performance is evaluated via simulations and discussing germane numerical results.
\end{abstract}


\providelength{\AxesLineWidth}       \setlength{\AxesLineWidth}{0.5pt}
\providelength{\plotwidth}           \setlength{\plotwidth}{8cm}
\providelength{\LineWidth}           \setlength{\LineWidth}{0.7pt}
\providelength{\MarkerSize}          \setlength{\MarkerSize}{3pt}
\newrgbcolor{GridColor}{0.8 0.8 0.8}
\newrgbcolor{GridColor2}{0.5 0.5 0.5}

\section{Introduction}\label{sec:intro}
Mobile broadband services have been falling afoul of a perennially upsurged demand for radio resources during recent years. This upswing owes to the gigantic boom in mobile service subscribers' quantity as well as to the outgrowth of their generated traffic volume \cite{EricssonMobilityReport2013}. On the other hand, the migration of cellular network providers from offering a single service such as the Internet access to a multi-service framework, like multimedia telephony and mobile-TV \cite{QoS_3GPP}, along with the emergence and prevalence of smartphones hosting simultaneously running delay-tolerant and real-time applications with distinctive quality of service (QoS) requirements \cite{GhorbanzadehICNC2013} arise an urgency to dynamically provisioning various bit rates to the application traffic so as to elevate users' quality of experience (QoE) tightly bound to the subscriber churn \cite{QoS_3GPP}. As such, incorporating service differentiation mechanisms into resource allocation methods is a matter of 
high consequence. Inasmuch as applications' temporal usage percentage directly impacts the generated traffic volume and nature, e.g. the traffic elasticity, including the usage percentage as an application status differentiation in resource allocation schemes is worthwhile. Besides, cellular network providers capability to adopt a subscription-based differentiation \cite{QoS_3GPP}, wherein miscellaneous clients of an identical service receive differentiated subscription-based treatments (corporate vs. private, post-paid vs. pre-paid, and privileged vs. roaming users), can fine-tune resource allocation approaches. Henceforth, resource allocation modi operandi can accommodate diverse exigencies of present-day wireless networks conveying the hybrid traffic by accounting for all the aforementioned issues. Nonetheless, the majority of resource allocation proposals fizzle to address the aforesaid concerns collectively (section \ref{sec:related}).

This paper puts forward a novel convex utility proportional fairness maximization formulation for an optimal resource allocation in wireless networks and is outfitted with the subscriber, application status, and service differentiations parameterized respectively as user equipment (UE) subscription weights, application status weights, and application utility functions. The weights are supplied by network providers so that a foreground-running application such as a voice call attains a higher application status weight than do the background-running ones, e.g. an automatic application update process. Mobile subscribers of the system under our consideration can concurrently run multiple applications with their utility functions and statuses depending on the generated traffic nature and instantaneous usage percentage, respectively.

Moreover, casting the service differentiation under a utility proportional fairness policy prioritizes the real-time traffic over the delay-tolerant one, conducive to fulfilling QoS requirements. In addition to solving the formalized optimization problem analytically, we develop distributed and centralized solution procedures as computationally efficient algorithms excerpted from Lagrangians of the resource allocation's dual problems \cite{Boyd} and perform necessary simulations to validate leveraged methodologies. For the distributed case, the rate assignment process is realized in double stages which first optimally allocates UEs the Evolved Node B (eNB) resources via their mutual collaborations and then disseminate UE bandwidths to the running applications internally to the UEs in an optimum fashion. In contrast, the devised centralized routine allots hybrid application rates in a monolithic stage transacted in the cellular network provider side of the communications system.

\subsection{Related Work}\label{sec:related}
The resource allocation optimization research area has received a significant attention since the seminal network utility maximization study in \cite{kelly98ratecontrol} which allocated user rates through a utility proportional fairness maximization solved by the Lagrange multipliers \cite{Boyd}. Soon after, an iterative algorithm relying on the duality of the aforementioned resource allocation problem was proposed \cite{Low99optimizationflow}. Whilst the traffic in these early research works had an elastic nature common for wired communication systems and approximated by concave utility functions, the advent and prevalence of high-speed wireless networks have entailed an increased utilization of real-time applications whose utility functions grow non-concave \cite{Shenker95fundamentaldesign}. For instance, the utility  of a voice-over-IP (VoIP) can be represented as a step function whose utility is zero before a certain threshold rate and achieves $100\%$ for rates larger than the threshold. Another example 
is a video streaming application whose utility can be approximated as a sigmoidal function convex (concave) for rates below (above) its inflection point. As such, the methods presented in (\cite{kelly98ratecontrol,Low99optimizationflow}) incur the proceeding drawbacks: (a) Reaching optimal solutions for solely concave utility functions, they are inapplicable to the drastically escalating inelastic traffic volume of au courant networks; (b) Neither priority do they render to real-time applications with stringent QoS requirements, nor they reserve any attention for the application statuses, nor they look after subscribers' varied importance pivotal from a business standpoint.

Later, the authors in (\cite{Lee05non-convexoptimization,DL_PowerAllocation}) presented distributed rate allocation algorithms for multi-class service offerings based on concave and sigmoidal utility functions representing applications. Despite closely approximating optimal solution, involved methods dropped users to maximize the system utility, so they could not guarantee a minimal QoS. An effort by the authors in (\cite{AbdelhadiCNC2014, AbdelhadiPIMRC2013, AbdelhadiMobicom2013}) proposed a utility proportional fairness resource allocation, for users of a single-carrier communication network, cast as a convex problem with logarithmic and sigmoidal utility functions respectively modelling delay-tolerant and real-time applications. Although their schemes prioritized the real-time applications over the delay-tolerant ones, they neither contemplated the application status or user differentiation concepts, nor regarded the hybrid traffic prevailing in modern networks.

In \cite{RebeccaThesis}, the author considered a weighted aggregation of logarithmic and sigmoidal utilities approximated to the nearest concave utility function via a minimum mean-squared error measure inside UEs. The approximate utility function solved the rate allocation optimization through a variation of the conventional distributed resource allocation approach in \cite{kelly98ratecontrol} such that rate assignments essentially estimated optimal ones. However, the rate were only approximations and no consideration was given to user or application priorities. This work was extended by Shajaiah et. al. (\cite{ShajaiahICNC2014,ShajaiahMILCOM2013}) to allow for the application of the resource allocation in a multi-carrier network in public safety. The authors in (\cite{ShajaiahPIMRC2014,ShajaiahCCS2014}) considered a similar multicarrier optimal resource allocation aware of the subscriber priorities. However, no attention was rendered to the temporal changes in the application usage or UE quantities. In (\cite{DBLP:conf/globecom/TychogiorgosGL11}), the authors adopted a non-convex optimization formulation to maximize the system utility in wireless networks consisting of applications with logarithmic and sigmoidal utility functions. A distributed process was employed to obtain the rates under a zero duality gap; but, the algorithm did not converge for a positive duality gap leading to compounding a heuristic to ensure the network stability.

In other studies, the authors of \cite{DBLP:conf/qosip/Harks05} created a utility max-min fairness resource allocation for the hybrid traffic sharing a single path in a communications network. Similarly, \cite{UtilityFairness} presented a utility proportional fairness optimization for the high signal-to-interference-plus-noise ratio wireless networks using a utility max-min architecture, contrasted against the traditional proportional fairness algorithms \cite{utility_fair} and provided a closed-form solution that refrained from network oscillations. However, neither methods cared for any traffic or user priorities in assigning the spectrum. In (\cite{GhorbanzadehICNC2015,Erpek2015}), the authors developed a utility proportional fairness resource block allocation in wireless networks as an integer optimization problem. They initially obtained the continuous optimal rates and then took on a boundary mapping technique to extracted a pool of valid resource blocks tantamount to inferred optimal continuous rates, 
albeit neither hybrid traffic, nor application status, nor user importance was taken into the equation. In a similar work \cite{GhorbanzadehMILCOM2014}, the authors organized a utility proportional fairness optimization which allocated optimal UE rates in a cellular infrastructure coexistent with radars by leveraging the Lagrange multipliers. Finally, \cite{Tao2008} presented a subcarrier allocation in orthogonal frequency division multiplexed systems concentrating on delay constrained data and used network delay models \cite{GhorbanzadehICC2013} for the subcarrier assignment. And last but not the least, \cite{AbdelhadiarXiv2014_2} developed a location/time/context-aware source allocation in cellular networks; however, they did not consider the temporal changes in the application usage percentage, the number of UEs, or subscribers' priority.

\subsection{Contributions}\label{sec:contribution}
In brevity, contributions of the current manuscript proceed as such.
\begin{itemize}
\item We formulate resource allocation optimizations with centralized and distributed architectures for cellular communication systems subsuming smartphones generating a hybrid traffic of elastic and inelastic data flows respectively stemmed from concurrently running delay-tolerant and real-time applications applications mathematically modelled as logarithmic and sigmoidal utility functions in that order.
\item We prove that the proposed resource allocation optimization problems are convex, have tractable global optimal solutions (the rate assignments are optimal), render bandwidth assignment priorities to real-time applications due to their reliance on the utility proportional fairness framework, and eschew from dropping users hereby a minimum QoS is warranted.
\item We adopt a two-stage algorithm for the distributed approach to optimally assign rates for UEs externally and for running applications internally.
\item We derive a robust one-stage algorithm for the centralized scheme to optimally assign rates to the running applications externally to the UEs and prove that the mathematical equivalence of the two-stage and one-stage concocts.
\end{itemize}

\subsection{Organization}\label{sec:organization}
The remainder of this paper is organized as follows. Section \ref{sec:related} surveys the topical resource allocation literature in brevity. Section \ref{sec:Problem_formulation} presents the formulation for the centralized and distributed resource allocation optimizations problems. Section \ref{sec:global_optimal} proves the existence of global optimal solutions for the optimization problems devised in section \ref{sec:Problem_formulation}. Section \ref{sec:EURA_Dual_algorithm} puts forward solution algorithms for the optimization problems. Section \ref{sec:equivalence} proves the mathematical equivalence between the distributed and centralized algorithms. Section \ref{sec:robust_alg} illustrates the distributed robust rate allocation algorithm to refrain from rate fluctuations. Section \ref{sec:one_stage_alg} provides with a solution algorithm for the one-stage centralized rate allocation algorithm. Section \ref{sec:sim} discusses simulation setup and develops quantitative results along with their 
analysis for the implementation of the proposed resource assignment schemes. And, section \ref{sec:concl} concludes the paper.

\section{Problem Formulation}\label{sec:Problem_formulation}
The objective is to determine optimal rates that hybrid-traffic-carrying cellular communications systems should be allocating to their UE applications so as to dynamically ensure as such: 1) Real-time applications are rendered priority over delay-tolerant ones. 2) no user is dropped 3) Applications temporal usage is accounted for. 4) Subscription-based treatments is honored. We assume each UE contains multiple simultaneously running real-time and delay-tolerant applications, mathematically represented by sigmoidal and logarithmic utility functions as shown in section \ref{sec:utilities}.

\subsection{Applications Utility functions}\label{sec:utilities}
Utility function have been used in a wide variety of research works to model some representative characteristic of the system. For instance, \cite{AbdelhadiarXiv2014_1} leveraged utility functions to model the modulation schemes in a power allocation problem. In this paper, an application performance satisfaction as a function of its allocated rates is referred to as a utility function, denoted as $U(r)$ for the rate $r$, and have the following properties \cite{Shenker95fundamentaldesign}.

\begin{itemize}
\item $U(0) = 0$ and $U(r)$ is an increasing function $r$ .
\item $U(r)$ is twice differentiable in $r$ and bounded above.
\end{itemize}

The first statement of the former property implies the nonnegativity of the utility functions which is expected since they represent application performance satisfaction percentage, whereas its second statement reveals that the more assigned rate, the higher the application performance satisfaction. On the flip side, the latter property indicates the continuity of the utility functions. Hybrid traffic consists of elastic and inelastic traffic streams sprung from respectively delay-tolerant and real-time applications whose utilities are conductively modelled by correspondingly normalized logarithmic and sigmoidal utility functions in equations (\ref{eqn:sigmoid}) and (\ref{eqn:log}) in that order (\cite{DL_PowerAllocation}).

\begin{equation}\label{eqn:sigmoid}
    U(r) = c\Big(\frac{1}{1+e^{-a(r-b)}}-d\Big)
\end{equation}

Here, $c = \frac{1+e^{ab}}{e^{ab}}$ and $d = \frac{1}{1+e^{ab}}$. It can be easily verified that $U(0) = 0$ and $U(\infty) = 1$, where the former is one of the previously mentioned utility function properties and the latter indicates that an infinite resource assignment ensues $100\%$ satisfaction. Furthermore, it is easily derivable that the inflection point of equation (\ref{eqn:sigmoid}) occurs at $r = r^{\text{inf}} = b$, where the superscript "inf" stands for infliction.

\begin{equation}\label{eqn:log}
    U(r) = \frac{\log(1+kr)}{\log(1+kr^{\text{max}})}
\end{equation}

Here, $r^{\text{max}}$ is the maximum rate at which the application QoS is satisfied in full ($100\%$ utility percentage) and $k$ is the utility function increase with augmenting the allocated rate $r$. It can be easily checked that $U(0) = 0$ and $U(r^{\text{max}}) = 1$., where the former is again the basic property of the utility functions and the latter implies that a $100\%$ QoS satisfaction occurs at $r = r^{max}$. Moreover, the inflection point of normalized logarithmic function is at $r = r^{\text{inf}} = 0$. For the sake of illustration, the utility functions with the parameters according to Table \ref{table:parameters} are plotted in Figure \ref{fig:sim:app_utilities}, from which we can observe that the sigmoid utility functions gain a slight QoS satisfaction only after the allocated rates surpass the inflection points of the utilities whereas the logarithmic ones obtain some QoS fulfillment even for a minuscule assigned bandwidth. These behaviors make sigmoidal and logarithmic utility functions 
suitable for modeling real-time and delay-tolerant applications respectively, and the germane mathematical analyses appear in (\cite{DL_PowerAllocation,UtilityFairness}) in nuance.

Next, section \ref{sec:sys_model} concocts the system model for the rate allocation problem proposed in this article.

\subsection{System Model}\label{sec:sys_model}
To present the system, with no loss of generality, we concentrate on a cellular network's single cell, which subsumes an eNB covering $M$ UEs (here $M = 6$) depicted in Figure \ref{fig:system_model}, where each UE concurrently runs delay-tolerant and real-time applications represented respectively by the logarithmic and sigmoidal utility functions in section \ref{sec:utilities}. The rate assigned by the eNB to the $i^{th}$ UE is denoted as $r_{i}$ and the UE's aggregated utility function is shown as $V_i(r_{i})$, which we relate it to the UE application utilities accordingly to the equation (\ref{eqn:utility_agg}) below.

\begin{equation}\label{eqn:utility_agg}
V_i(r_{i}) = \prod_{j=1}^{N_i}U_{ij}^{\alpha_{ij}}(r_{ij})
\end{equation}

Here, $r_{ij}$, $U_{ij}(r_{ij})$, and $\alpha_{ij}$ respectively represent the rate allocation, application utility function, and application usage percentage of the $j^{th}$ application running on the $i^{th}$ UE. Hence, we can write $\sum_{j=1}^{N_i}{\alpha_{ij}} = 1$ and $r_i = \sum_{j=1}^{N_i}{r_{ij}}$, where, for the $i^{th}$ UE, $N_i$ is the number of coevally running applications and $r_i$ presents the bandwidth allotment by the eNB. The former of the afore-written equations states the fact that the addition of the $i^{th}$ UE's application usage percentages proves $100\%$ usage percentage, and the letter one implies that the the $i^{th}$ UE rate is the augmentation of all its $N_i$ applications resources assignments.

We resort to a centralized and a distributed approach, illustrated in sections \ref{sec:one_stage} and \ref{sec:two_stage} respectively, to disseminate resources to the applications of UEs with the aggregated utility as in equation (\ref{eqn:utility_agg}).

\begin{figure}[!htb]
\begin{center}
\includegraphics[width=3.5in]{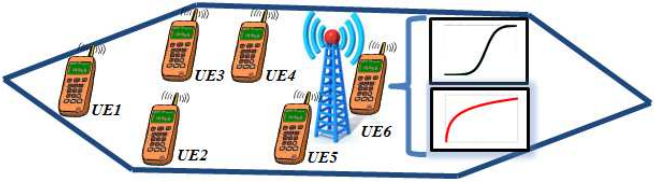}
\end{center}
\caption{\footnotesize{System Model: Single cell, within the cellular network, with an eNB covering $M = 6$ UEs each with simultaneously running delay-tolerant and relay-time applications represented by logarithmic and sigmoidal utility functions respectively.}}\label{fig:system_model}
\end{figure}

\subsection{Centralized Optimization}\label{sec:one_stage}
We develop a rate allocation optimization problem that assigns the application resources directly by the eNB in a singular stage. The basic germane formulation is illustrated in equation (\ref{eqn:opt_multiapp}).

\begin{equation}\label{eqn:opt_multiapp}
\begin{aligned}
& \underset{\textbf{r}}{\text{max}}
& & \prod_{i=1}^{M}\Big(\prod_{j=1}^{N_i}U_{ij}^{\alpha_{ij}}(r_{ij})\Big)^{\beta_{i}} \\
& \text{subject to}
& & \sum_{i=1}^{M}\sum_{j=1}^{N_i}r_{ij} \leq R,\\
& & &  r_{ij} \geq 0, \;\; i = 1,2, ...,M,\;\;j = 1,2,...,N_i\\
\end{aligned}
\end{equation}

Here, for $M$ UEs covered by an eNB, $\textbf{r} = [r_1,r_2,...,r_M]$ is the UE allocated rate vector, $R$ is the maximum available resources at the eNB, and $\beta_i$ is a subscription-dependent weight for the $i^{th}$ UE. Section \ref{sec:one_stage_global_optimal} proves the convexity and tractable optimal solvability of the aforementioned optimization problem, whose solution procedure is concocted as an algorithm in section \ref{sec:one_stage_alg}.

Next, section \ref{sec:two_stage} introduces a distributed approach to the resource allocation problem.

\subsection{Distributed Optimization}\label{sec:two_stage}
We subdivide the optimization problem (\ref{eqn:opt_multiapp}) into two simpler optimizations solved separately. The first optimization concerns with the UE rate allocation by the eNB via collaborations between the eNB and pertinent UEs hereby the optimization is referred to as external UE resource allocation (EURA). On the contrary, the second optimization wells up from distributing the application rates by the host UEs, performed internally to the UEs and is named the internal UE rate allocation (IURA). The solutions for the EURA and IURA optimization is laid as algorithms presented in section \ref{sec:EURA_Dual_algorithm}. The EURA formulation is explained below.

\subsubsection{EURA Optimization Problem}\label{sec:UE_alloc_opt}
EURA optimization, solved collaboratively amongst UEs and their eNB, can be written as equation (\ref{eqn:opt_sub1}), where $V_i(r_i) = \prod_{j=1}^{N_i}U_{ij}^{\alpha_{ij}}(r_{ij})$ is the $i^{th}$ UE aggregated utility function expressed in equation (\ref{eqn:opt_multiapp}), $\textbf{r} = [r_1,r_2,...,r_M]$ is the UE rate vector whose $i^{th}$ component represents the rate assigned by the eNB to the $i^{th}$ UE, and $M$ is the number of UEs covered by the eNB. In section \ref{sec:EURA_global_optimal}, we prove the convexity and tractable optimal solvability of the optimization problem in equation (\ref{eqn:opt_sub1}) exists and present the algorithm to solve this problem in section \ref{sec:global_optimal}.

\begin{equation}\label{eqn:opt_sub1}
\begin{aligned}
& \underset{\textbf{r}}{\text{max}}
& & \prod_{i=1}^{M}V_i^{\beta_{i}}(r_{i}) \\
& \text{subject to}
& & \sum_{i=1}^{M}r_{i} \leq R,\\
& & &  r_{i} \geq 0, \;\;\;\;\; i = 1,2, ...,M.
\end{aligned}
\end{equation}

\subsubsection{IURA Optimization Problem}\label{sec:app_alloc_opt}
IURA optimization problem, solved internally in each UE, can be written as equation (\ref{eqn:opt_sub2}) for the $i^{th}$ UE with $i \in \{1,2,...M\}$, where $\textbf{r}_i = [r_{i1},r_{i2},...,r_{iN_i}]$ is the application rate allocation vector such that its $j^{th}$ component indicates the bandwidth allotted by the $i^{th}$ UE to its $j^{th}$ application, $r_{i}^{\text{opt}}$ is the $i^{th}$ UE rate allocated by eNB via solving the EURA optimization in equation (\ref{eqn:opt_sub1}), and $N_i$ is the number of applications of the $i^{th}$ UE. Superscript "opt" indicates the optimality of the UE rates which will be proved in section \ref{sec:EURA_global_optimal}. Besides, section \ref{sec:IURA_global_optimal} proves that there exists a tractable global optimal solution to the IURA optimization problem in equation (\ref{eqn:opt_sub2}) and section \ref{sec:IURA_alg} provides the solving algorithm thereof.

\begin{equation}\label{eqn:opt_sub2}
\begin{aligned}
& \underset{\textbf{r}_i}{\text{max}}
& & \prod_{j=1}^{N_i}U_{ij}^{\alpha_{ij}}(r_{ij}) \\
& \text{subject to}
& & \sum_{j=1}^{N_i}r_{ij} \leq r_{i}^{\text{opt}},\\
& & & r_{ij} \geq 0, \;\;\;\;\; j = 1,2, ...,N_i.
\end{aligned}
\end{equation}

Next, section \ref{sec:global_optimal} proves the convexity of the EURA and IURA optimization problems.

\section{Existence of a Global Optimal Solution}\label{sec:global_optimal}
This section proves the existence of optimal solutions for the centralized and distributed resource allocations developed in sections \ref{sec:two_stage} and section \ref{sec:one_stage}, respectively.

\subsection{EURA Global Optimal Solution}\label{sec:EURA_global_optimal}
Strictly increasing nature of logarithms yields in an equivalent EURA objective function $\arg \underset{\textbf{r}} \max \sum_{i=1}^{M}\beta_{i} \log(V_i(r_{i}))$, stemmed from equation (\ref{eqn:opt_sub1}), reformulated and referred to as respectively equation \ref{eqn:opt_sub1_log} and log-EURA problem, for which the lemma \ref{lem:concavity} is conceivable.

\begin{equation}\label{eqn:opt_sub1_log}
\begin{aligned}
& \underset{\textbf{r}}{\text{max}}
& & \sum_{i=1}^{M}\beta_{i} \log(V_i(r_{i})) \\
& \text{subject to}
& & \sum_{i=1}^{M}r_{i} \leq R,\\
& & &  r_{i} \geq 0, \;\;\;\;\; i = 1,2, ...,M.
\end{aligned}
\end{equation}

\begin{lem}\label{lem:concavity}
The aggregated utility natural logarithm $\log(V_{i}(r_{i}))$ is strictly concave.
\end{lem}

\begin{proof}
From equation (\ref{eqn:utility_agg}), we can write $\log V_i(r_{i}) = \sum_{j=1}^{N_i}{\alpha_{ij}}\log U_{ij}(r_{ij})$ where $U_{ij}(r_{ij}) > 0$ in accordance with section \ref{sec:Problem_formulation} utility function properties. Also, logarithmic utilities (equation (\ref{eqn:log})) concavity stems out $U'_{ij}(r_{ij}) = \frac{dU_{ij}(r_{ij})}{dr_{ij}} > 0$ and $U''_i(r_{ij}) = \frac{d^2U_{ij}(r_{ij})}{dr_{ij}^2} < 0$, resulting in $\frac{d\log(U_{ij}(r_{ij}))}{dr_{ij}} =  \frac{U'_{ij}(r_{ij})}{U_{ij}(r_{ij})} > 0$ due to $U_{ij}(r_{ij}) > 0$ and $U'_{ij}(r_{ij}) > 0$ and in $\frac{d^2\log(U_{ij}(r_{ij}))}{dr_{ij}^2} =  \frac{U''_{ij}(r_{ij})U_{ij}(r_{ij})-U'^2_{ij}(r_{ij})}{U^2_{ij}(r_{ij})} < 0$ due to $U''_{ij}(r_{ij}) < 0$. Thus, the logarithmic utility natural logarithm is strictly concave. On the flip side, for a sigmoidal utility $U_{ij}(r_{ij})$ (equation (\ref{eqn:sigmoid})) with $0<r_{ij}<R$, we have the following inequalities amongst which the first owes to the sigmoidal function's 
continuity and $0 \leq U_{ij}(r_{ij}) < 1$ and the rest are utter algebraic manipulation of the first one.

\begin{equation*}\label{eqn:sigmoid_bound}
\begin{aligned}
0&<c_{ij}\Big(\frac{1}{1+e^{-a_{ij}(r_{ij}-b_{ij})}}-d_{ij}\Big)<1\\
d_{ij}&<\frac{1}{1+e^{-a_{ij}(r_{ij}-b_{ij})}}<\frac{1+c_{ij}d_{ij}}{c_{ij}}\\
\frac{1}{d_{ij}}&>{1+e^{-a_{ij}(r_{ij}-b_{ij})}}>\frac{c_{ij}}{1+c_{ij}d_{ij}}\\
0&<1-d_{ij}({1+e^{-a_{ij}(r_{ij}-b_{ij})}})<\frac{1}{1+c_{ij}d_{ij}}\\
\end{aligned}
\end{equation*}

For $0 < r_{ij} < R$, we have the following inequalities, of which the first results from first additive's denominator positivity in addition to the formerly derived statement $0 < 1 - d_{ij}(1 + e^{-a_{ij}(r_{ij} - b{ij})}) < \frac{1}{1 + c_{ij}d_{ij}}$ as well as other constituents' positivity and the last one is verifiable by investigating its terms algebraically. Hence, the sigmoidal utility natural logarithm is strictly concave. As such, the applications utility functions $U_{ij}(r_{ij}) > 0$ of the system model (equation (\ref{eqn:utility_agg})) have strictly concave natural logarithms, meaning that the aggregated utility $\log V_i(r_{i}) = \sum_{j=1}^{N_i}{\alpha_{ij}}\log U_{ij}(r_{ij})$ is strictly concave.

\begin{equation}\label{eqn:sigmoid_derivative}
\begin{aligned}
\frac{d}{dr_{ij}}\log U_{ij}(r_{ij}) =& \frac{a_{ij}d_{ij} e^{-a_{ij}(r_{ij}-b_{ij})}}{1-d_{ij}(1+e^{-a_{ij}(r_{ij}-b_{ij})})} \\
\;\;\;&  + \frac{a_{ij}e^{-a_{ij}(r_{ij}-b_{ij})}}{(1+e^{-a_{ij}(r_{ij}-b_{ij})})}>0\\
\frac{d^2}{dr_{ij}^2}\log U_{ij}(r_{ij}) =& \frac{-a_{ij}^2d_{ij}e^{-a_{ij}(r_{ij}-b_{ij})}}{c_{ij}\Big(1-d_{ij}(1+e^{-a(r_{ij}-b_{ij})})\Big)^2} \\
\;\;\;&  + \frac{-a_{ij}^2e^{-a_{ij}(r_{ij}-b_{ij})}}{(1+e^{-a_{ij}(r_{ij}-b_{ij})})^2} < 0\\
\end{aligned}
\end{equation}
\end{proof}

Next, theorem \ref{thm:EURA_global_soln} proves the EURA optimization convexity.

\begin{thm}\label{thm:EURA_global_soln}
The EURA optimization problem in equation (\ref{eqn:opt_sub1}) is convex and has a unique tractable global optimal solution.
\end{thm}

\begin{proof}
The aggregated utility concavity (lemma \ref{lem:concavity}) concludes the log-EURA (equation (\ref{eqn:opt_sub1_log})) convexity \cite{Boyd2004}, which in turn proves the EURA problem (equation (\ref{eqn:opt_sub1})) convexity due to their objective functions equivalence. There exists a unique tractable global optimal solution for a convex optimization \cite{Boyd2004}.
\end{proof}

Section \ref{sec:IURA_global_optimal} explores the IURA optimization convexity.

\subsection{IURA Global Optimal Solution}\label{sec:IURA_global_optimal}
The IURA objective function $\prod_{j=1}^{N_i}U_{ij}^{\alpha_{ij}}(r_{ij})$ corresponds to $\sum_{j=1}^{N_i}\alpha_{ij} \log(U_{ij}(r_{ij}))$, so equation (\ref{eqn:opt_sub2}) can be reformulated as equation \ref{eqn:opt_sub2_log}, referred to as the log-IURA problem for which corollary \ref{cor:iura_conv} is conceivable.

\begin{equation}\label{eqn:opt_sub2_log}
\begin{aligned}
& \underset{\textbf{r}_i}{\text{max}}
& & \sum_{j=1}^{N_i}\alpha_{ij} \log U_{ij}(r_{ij}) \\
& \text{subject to}
& & \sum_{i=1}^{N_i}r_{ij} \leq r_{i}^{\text{opt}},\\
& & & r_{ij} \geq 0, \;\;\;\;\; j = 1,2, ...,N_i.
\end{aligned}
\end{equation}

\begin{cor}\label{cor:iura_conv}
The IURA optimization problem in equation (\ref{eqn:opt_sub2}) is convex and has a unique tractable global optimal solution.
\end{cor}

\begin{proof}
Substantiating lemma \ref{lem:concavity} was concomitant with proving the application utility natural logarithm concavity, which ascertains the convexity of the log-IURA (equation (\ref{eqn:opt_sub2_log})) \cite{Boyd2004}, yielding in the convexity of its equivalent IURA optimization (equation (\ref{eqn:opt_sub2})) and existence of a tractable global optimal solution \cite{Boyd2004}.
\end{proof}

Theorem \ref{thm:EURA_global_soln} and corollary \ref{cor:iura_conv} indicate that the distributed optimization in section \ref{sec:two_stage} assigns rates optimally. Next, section \ref{sec:one_stage_global_optimal} analyzes the centralized optimization convexity.

\subsection{Centralized Rate Allocation Global Optimal Solution}\label{sec:one_stage_global_optimal}
The centralized optimization objective function $\prod_{i=1}^{M}\Big(\prod_{j=1}^{N_i}U_{ij}^{\alpha_{ij}}(r_{ij})\Big)^{\beta_{i}}$ corresponds to $\sum_{i=1}^{M}{\beta_{i}}\sum_{j=1}^{N_i}{\alpha_{ij}}\log U_{ij}(r_{ij})$, reformulating equation (\ref{eqn:opt_multiapp}) as equation \ref{eqn:opt_multiapp_log}, referred to as the log-centralized problem, for which corollary \ref{cor:centr_multiapp_conv} is conceivable.

\begin{equation}\label{eqn:opt_multiapp_log}
\begin{aligned}
& \underset{\textbf{r}}{\text{max}}
& & \sum_{i=1}^{M}{\beta_{i}}\sum_{j=1}^{N_i}{\alpha_{ij}}\log U_{ij}(r_{ij}) \\
& \text{subject to}
& & \sum_{i=1}^{M}\sum_{j=1}^{N_i}r_{ij} \leq R,\\
& & &  r_{ij} \geq 0, \;\; i = 1,2, ...,M,\;j = 1,2,...,N_i\\
\end{aligned}
\end{equation}

\begin{cor}\label{cor:centr_multiapp_conv}
The centralized optimization problem in equation (\ref{eqn:opt_multiapp}) is convex and has a unique tractable global optimal solution.
\end{cor}

\begin{proof}
Substantiating lemma \ref{lem:concavity} was concomitant with proving the application utility natural logarithm concavity, which entails the the log-centralized optimization (equation (\ref{eqn:opt_multiapp_log})) convexity \cite{Boyd2004}, ensuing the convexity of its tantamount centralized optimization (equation (\ref{eqn:opt_multiapp})) and existence of a tractable global optimal solution \cite{Boyd2004}.
\end{proof}

Next, section \ref{sec:EURA_Dual_algorithm} solves the distributed optimization problem presented in section \ref{sec:two_stage}.

\section{EURA Algorithm and Drawback}\label{sec:EURA_Dual_algorithm}
Here, we deploy the distriduality for convex optimization problems to solving them efficiently, similar to \cite{kelly98ratecontrol,Low99optimizationflow}. What proceeds is such an application of the duality to EURA and IURA constituents of the distributed rate allocation problem as well as to the centralized resource assignment problem. We present the EURA algorithm in section \ref{sec:distributed-alg-subsec} below.

\subsection{EURA Algorithm}\label{sec:distributed-alg-subsec}
The log-EURA problem (\ref{eqn:opt_sub1_log}) can be solved by converting it to its dual problem, similar to \cite{kelly98ratecontrol,Low99optimizationflow}. We define the Lagrangian as equation (\ref{eqn:lagrangian}).

\begin{equation}\label{eqn:lagrangian}
\begin{aligned}
L(\textbf{r},p) = & \sum_{i=1}^{M}{\log(V_{i}(r_{i}))-p(\sum_{i=1}^{M}r_{i} + \sum_{i=1}^{M}z_{i} - R)}\\
                = &  \sum_{i=1}^{M}\Big({\log(V_{i}(r_{i}))-pr_{i}\Big)} + p(R-\sum_{i=1}^{M}z_{i})\\
                = &  \sum_{i=1}^{M}L_i(r_{i},p) + p(R-\sum_{i=1}^{M}z_{i})\\
\end{aligned}
\end{equation}

where $z_{i}\geq 0$ is the slack variable and $p$ is Lagrange multiplier or the shadow price (price per unit bandwidth for all the $M$ channels). Therefore, the $i^{th}$ UE bid for bandwidth can be written as $w_i = p r_{i}$, where $\sum_{i=1}^{M}w_i = p \sum_{i=1}^{M}r_{i}$. The first term in equation (\ref{eqn:lagrangian}) is separable in $r_{i}$, so we have $\underset{\textbf{r}}\max \sum_{i=1}^{M}({\log(V_{i}(r_{i}))-pr_{i})} = \sum_{i=1}^{M}\underset{{r_{i}}}\max({\log(V_{i}(r_{i}))-pr_{i})}$ and the dual problem objective function can be written as equation (\ref{eqn:dual_obj_fn}).

\begin{equation}\label{eqn:dual_obj_fn}
\begin{aligned}
D(p) = & \underset{{\textbf{r}}}\max \:L(\textbf{r},p) \\
= &\sum_{i=1}^{M}\underset{{r_{i}}}\max\Big({\log(V_{i}(r_{i}))-pr_{i}\Big)} + p(R-\sum_{i=1}^{M}z_{i})\\
= &\sum_{i=1}^{M}\underset{{r_{i}}}\max (L_i(r_{i},p)) + p(R - \sum_{i=1}^{M}z_{i})
\end{aligned}
\end{equation}

Thus, the dual problem is formulated as equation (\ref{eqn:dual_problem}).

\begin{equation}\label{eqn:dual_problem}
\begin{aligned}
& \underset{{p}}{\text{min}}
& & D(p) \\
& \text{subject to}
& & p \geq 0.
\end{aligned}
\end{equation}

Leveraging the method of Lagrange multiplier, we have:

\begin{equation}\label{eqn:dual_max}
\frac{\partial D(p)}{\partial p} =  R-\sum_{i=1}^{M}r_{i} - \sum_{i=1}^{M}z_{i} = 0
\end{equation}

Substituting by $\sum_{i=1}^{M}w_i = p \sum_{i=1}^{M}r_{i}$, we have equation (\ref{eqn:dual_new_obj}), minimized to $p = \frac{\sum_{i=1}^{M}w_i}{R}$ at $\sum_{i=1}^{M}z_{i} = 0$ where $w_i = p r_{i}$ is transmitted by the $i^{th}$ UE to the eNB.

\begin{equation}\label{eqn:dual_new_obj}
p = \frac{\sum_{i=1}^{M}w_i}{R-\sum_{i=1}^{M}z_{i}}
\end{equation}

As such, we divide the log-EURA problem (\ref{eqn:opt_sub1_log}) into simpler optimizations at the eNB (eNB EURA problem) and UEs (UE EURA problem), respectively equations (\ref{eqn:opt_prob_fairness_eNodeB}) and (\ref{eqn:opt_prob_fairness_UE}) whose solutions, guaranteeing the utility proportional fairness in equation (\ref{eqn:opt_sub1}), are summarized in Algorithms \ref{alg:eNodeB_distributed} and \ref{alg:UE_distributed} in that order.

\begin{equation}\label{eqn:opt_prob_fairness_UE}
\begin{aligned}
& \underset{{r_{i}}}{\text{max}}
& & \log V_{i}(r_{i}) - p r_{i} \\
& \text{subject to}
& & p \geq 0\\
& & &  r_{i} \geq 0, \;\;\;\;\; i = 1,2, ...,M.
\end{aligned}
\end{equation}

During the execution of the aforesaid algorithms, starting with $w_i(0) = 0$, the $i^{th}$ UE, transmits an initial bid $w_i(1)$ to the eNB, which in turn subtracts the latterly received bid $w_i(n)$ and the formerly received one $w_i(n-1)$ and ceases the procedure if the difference is less than a threshold $\delta$; Otherwise, it computes and sends a shadow price $p(n) = \frac{\sum_{i=1}^{M}w_i(n)}{R}$ to is covered UEs. The $i^{th}$ UE extracts its rate $r_i(n)$ from the received $p(n)$ such that $\log V_i(r_i) - p(n)r_i$ is maximized. The rate $r_i(n)$ is employed to estimate the new bid $w_i(n)=p(n) r_{i}(n)$, transmitted to the eNB. This routine repeats until the bid difference $|w_i(n) - w_i(n-1)|$ falls below the threshold $\delta$.

\begin{equation}\label{eqn:opt_prob_fairness_eNodeB}
\begin{aligned}
& \underset{p}{\text{min}}
& & D(p) \\
& \text{subject to}
& & p \geq 0.\\
\end{aligned}
\end{equation}

The solution $r_i(n)$ of the $i^{th}$ UE EURA optimization $r_{i}(n) = \arg \underset{r_i}\max \Big(\log V_i(r_i) - p(n)r_i\Big)$ in Algorithm \ref{alg:UE_distributed} essentially solves the equation $\frac{\partial \log V_i(r_i)}{\partial r_i} = p(n)$, algebraically the Lagrange multiplier solution for equation (\ref{eqn:opt_prob_fairness_UE}) and geometrically the intersection point of the horizontal line $y = p(n)$ with the curve $y = \frac{\partial \log V_i(r_i)}{\partial r_i}$.

\begin{algorithm}
\caption{UE EURA Optimization}\label{alg:UE_distributed}
\begin{algorithmic}
\STATE {Send initial bid $w_i(1)$ to eNB.}
\LOOP
	\STATE {Receive shadow price $p(n)$ from eNB.}
	\IF {STOP from eNB} %
		\STATE {Calculate allocated rate $r_i ^{\text{opt}}=\frac{w_i(n)}{p(n)}$.}
		\STATE {STOP}
	\ELSE
		\STATE {Solve $r_{i}(n) = \arg \underset{r_i}\max \Big(\log V_i(r_i) - p(n)r_i\Big)$.}
		\STATE {Send new bid $w_i (n)= p(n) r_{i}(n)$ to eNB.}
	\ENDIF
\ENDLOOP
\end{algorithmic}
\end{algorithm}

\begin{algorithm}
\caption{eNB EURA Optimization}\label{alg:eNodeB_distributed}
\begin{algorithmic}
\LOOP
	\STATE {Receive bids $w_i(n)$ from UEs.}
	\COMMENT{Let $w_i(0) = 1\:\:\forall i$}
	\IF {$|w_i(n) -w_i(n-1)|<\delta  \:\:\forall i$} %
   		\STATE {Allocate rates, $r_{i}^{\text{opt}}=\frac{w_i(n)}{p(n)}$ to user $i$.}
   		\STATE {STOP}
	\ELSE
		\STATE {Calculate $p(n) = \frac{\sum_{i=1}^{M}w_i(n)}{R}$.}
		\STATE {Send new shadow price $p(n)$ to all UEs.}
	\ENDIF
\ENDLOOP
\end{algorithmic}
\end{algorithm}

A convergence analysis of the EURA algorithms and its resultant snags are discussed in section \ref{sec:conv_analy}.

\subsection{EURA Convergence Analysis}\label{sec:conv_analy}
To commence analyzing the EURA Algorithms \ref{alg:UE_distributed} and \ref{alg:eNodeB_distributed}, lemma \ref{lem:slope_curve} is envisaged.

\begin{lem}\label{lem:slope_curve}
The aggregated utility function $V_{i}(r_{i})$,  the slope curvature function $\frac{\partial \log V_{i}(r_{i})}{\partial r_{i}}$ has inflection points at $r_{i} = r_{ij}^{s} \approx r_{ij}^{\text{inf}}$ for $j^{th}$ application utility function $U_{ij}$ and is convex for $r_{ij} > \underset{j}\max \:r_{ij}^{s}$.
\end{lem}

\begin{proof}
For the $i^{th}$ UE aggregated utility $V_{i}(r_{i})$, let $S_{i}(r_{i}) = \frac{\partial \log V_{i}(r_{i})}{\partial r_{i}}$ be the aggregated utility slope curvature function, $S_{ij}(r_{ij}) = \frac{\partial \log U_{ij}(r_{ij})}{\partial r_{ij}}$ be the $j^{th}$ application utility slope curvature function, and $N^{S}_i$ be the number of sigmoidal utilities. Taking the logarithm and derivative of both sides of the equation (\ref{eqn:utility_agg}) yields in equation (\ref{eqn:sum_slope}).

\begin{equation}
\begin{aligned}\label{eqn:sum_slope}
S_{i}(r_{i}) &= \frac{\partial \log V_{i}(r_{i})}{\partial r_{i}} = \frac{\partial}{\partial r_i}\sum_{j=1}^{N_i}{\alpha_{ij}}\log U_{ij}(r_{ij})\\
	      &= \sum_{j=1}^{N_i}{\alpha_{ij}}\frac{\partial \log U_{ij}(r_{ij})}{\partial r_{ij}}= \sum_{j=1}^{N_i}\alpha_{ij}S_{ij}(r_{ij})\\
	      &= \sum_{j=1}^{N^{S}_i}\alpha_{ij}S_{ij}(r_{ij})+\sum_{j=N^{S}_i+1}^{N_i}\alpha_{ij}S_{ij}(r_{ij})
\end{aligned}
\end{equation}

Taking the $1^{th}$ and $2^{nd}$ derivatives of equation (\ref{eqn:sum_slope}), we write:

\begin{equation}
\begin{aligned}\label{eqn:diff_slope}
\frac{\partial S_{i}}{\partial r_{i}} &= \sum_{j=1}^{N^{S}_i}\{\frac{-{\alpha_{ij}}a_{ij}^2d_{ij}e^{-a_{ij}(r_{ij}-b_{ij})}}{c_{ij}\Big(1-d_{ij}(1+e^{-a_{ij}(r_{ij}-b_{ij})})\Big)^2} \\
& + \frac{{\alpha_{ij}}a_{ij}^2e^{-a_{ij}(r_{ij}-b_{ij})}}{\Big(1+e^{-a_{ij}(r_{ij}-b_{ij})}\Big)^2}\} \\
& - \sum_{j=N^{S}_i+1}^{N_i}\{\frac{\alpha_{ij}k_{ij}^{2}}{(1+k_{ij}r^{max})\log(1+k_{ij}r_{ij})^{2}}\}\\
\end{aligned}
\end{equation}

\begin{equation}
\begin{aligned}\label{eqn:second_diff_slope}
\frac{\partial^2 S_{i}}{\partial r_{i}^2} &= \sum_{j=1}^{N^{S}_i}\{\frac{d_{ij}e^{-a_{ij}(r_{ij}-b_{ij})}(1-d_{ij}(1-e^{-a_{ij}(r_{ij}-b_{ij})}))}{c_{ij}\Big(1-d_{ij}(1+e^{-a_{ij}(r_{ij}-b_{ij})})\Big)^3} \\
& + \frac{e^{-a_{ij}(r_{ij}-b_{ij})}(1-e^{-a_{ij}(r_{ij}-b_{ij})})}{\Big(1+e^{-a_{ij}(r_{ij}-b_{ij})}\Big)^3}\}\times a_{ij}^3\alpha_{ij}\\
& - \sum_{j=N^{S}_i+1}^{N_i}\{\frac{\alpha_{ij}k^{2}_{ij}(\log(1+k_{ij}r_{ij})-1)}{(1+k_{ij}r_{ij})^{2}\log^{2}(1+k_{ij}r_{ij})}\}\\
\end{aligned}
\end{equation}

It is easy to show that $\forall\: r_{i},\:\:\frac{\partial S_{i}}{\partial r_{i}}<0$. Denoting the $1^{th}$ term of equation (\ref{eqn:diff_slope}) and $2^{nd}$ and $3^{rd}$ terms of equation (\ref{eqn:second_diff_slope}) as respectively $S^1_i$, $S^2_i$, and $S^3_i$ stems out equation set (\ref{eqn:slope_fn_terms}), for which properties in equation set (\ref{eqn:slope_fn_prop}) are considerable.

\begin{equation}\label{eqn:slope_fn_terms}
\left\{
\begin{array}{l l}
   S^1_i = \frac{\alpha_{ij}a_{ij}^3e^{a_{ij}b_{ij}}(e^{a_{ij}b_{ij}}+e^{-a_{ij}(r_{ij}-b_{ij})})}{(e^{a_{ij}b_{ij}}-e^{-a_{ij}(r_{ij}-b_{ij})})^3}\\
   S^2_i = \frac{a_{ij}^3\alpha_{ij}e^{-a_{ij}(r_{ij}-b_{ij})}(1-e^{-a_{ij}(r_{ij}-b_{ij})})}{\Big(1+e^{-a_{ij}(r_{ij}-b_{ij})}\Big)^3}\\
   S^3_i = \frac{\alpha_{ij}k^{2}_{ij}(\log(1+k_{ij}r_{ij})-1)}{(1+k_{ij}r_{ij})^{2}\log^{2}(1+k_{ij}r_{ij})}
\end{array} \right.
\end{equation}

From equation set (\ref{eqn:slope_fn_prop}), we observe that the slope curvature function $S_{i}$ has the inflection point $r_{i} = r_{ij}^{s} \approx b_{ij} = r_{ij}^{\text{inf}}$ and changes from a convex function close to the origin to a concave function before the inflection point at $r_{ij} = r_{ij}^{s}$ to a convex function after the inflection point.

\begin{equation}\label{eqn:slope_fn_prop}
\left\{
\begin{array}{l l}
   \lim_{r_{i} \rightarrow 0} S^1_i = \infty,\\
   \lim_{r_{i} \rightarrow b_{ij}} S^  1_i = 0 \:\:\text{for} \:\:b_{ij} 	\gg \frac{1}{a_{ij}}\:\:\forall \:\:j \\
   S^2_i (b_{ij}) = 0\\
   S^2_i (r_{ij}>b_{ij}) > 0\\
   S^2_i (r_{ij}<b_{ij}) < 0\\
   S^3_i (r_{ij}>0) > 0
\end{array} \right.
\end{equation}
\end{proof}

\begin{cor}\label{cor:sig_convergence}
If $\sum_{i=1}^{M}\underset{j}\max \:\:r_{ij}^{s} \ll R$, then Algorithms \ref{alg:UE_distributed} and \ref{alg:eNodeB_distributed} converge to the global optimal rates corresponding to the steady state shadow price $p_{ss}< \frac{a_{i_{\max}} d_{i_{\max}} }{1-d_{i_{\max}} }+\frac{a_{i_{\max} }}{2}$, where $i_{\max} = \arg \underset{i}\max\:\: r_{ij_{\max}}^{s}$ and $r_{ij_{\max}}^{s} = \underset{j}\max \:\:r_{ij}^{s}$.
\end{cor}

\begin{proof}
An essential step to reach the optimal solution in Algorithm (\ref{alg:UE_distributed}) is solving $r_{i}(n) = \arg \underset{r_{i}}\max \Big(\log V_{i}(r_{i}) - p(n)r_{i}\Big)$ using the Lagrange multipliers in equation (\ref{eqn:slope_equation}).

\begin{equation}\label{eqn:slope_equation}
\frac{\partial \log V_{i}(r_{i})}{\partial r_{i}} - p =  S_{i}(r_{i}) - p = 0.
\end{equation}

Furthermore, equation set (\ref{eqn:slope_fn_prop}) indicate that the slope curvature function $S_{i}(r_{i})$ is convex for $r_{i}  > \underset{j} \max r_{ij}^s \approx \underset{j} \max b_{ij}$. Similar to the analyses in \cite{kelly98ratecontrol,Low99optimizationflow}, the Algorithms \ref{alg:UE_distributed} and \ref{alg:eNodeB_distributed} are guaranteed to converge to the global optimal solution when the aggregated slope curvature function $S_{i}(r_{i})$ is in the convex region. Hence, the aggregated utility natural logarithm converges to the global optimal solution for $r_{i}  > \underset{j} \max r_{ij}^s \approx \underset{j} \max b_{ij}$. On the other hand, the sigmoidal utility function inflection point is at $r_{i}^{\text{inf}} = b_{ij}$. For $\sum_{i=1}^{M}\underset{j}\max \:\:r_{ij}^{s} \ll R$, Algorithms \ref{alg:UE_distributed} and \ref{alg:eNodeB_distributed} allocate rates $r_{ij}>b_{ij}$ for all users and since $S_{ij}(r_{ij})$ is convex for $r_{ij}>r_{ij}^s \approx b_{ij}$, the optimal 
solution can be achieved by the algorithms. Equation (\ref{eqn:slope_equation}) and convexity of $S_{ij}(r_{ij})$ for $r_{ij}  > r_{ij}^s \approx b_{ij}$ imply that $p_{ss}< S_{ij}(r_{ij} =\max b_{ij})$, where $S_{ij}(r_{ij} =\max b_{ij}) = \frac{a_{i_{\max}} d_{i_{\max}} }{1-d_{i_{\max}} }+\frac{a_{i_{\max} }}{2}$ and $i_{\max} = \arg \max_i b_{ij}$.
\end{proof}

\begin{cor}\label{cor:sig_fluctuate}
For $\sum_{i=1}^{M}\underset{j}\max \:\:r_{ij}^{s}>R$ and the global optimal shadow price $p_{ss} \approx \frac{a_{ij}d_{ij} e^{\frac{a_{ij}b_{ij}}{2}}}{1-d_{ij}(1+e^{\frac{a_{ij}b_{ij}}{2}})} + \frac{a_{ij}e^{\frac{a_{ij}b_{ij}}{2}}}{(1+e^{\frac{a_{ij}b_{ij}}{2}})}$, the solution by EURA Algorithms \ref{alg:UE_distributed} and \ref{alg:eNodeB_distributed} fluctuates about the global optimal solution.
\end{cor}

\begin{proof}
It follows from lemma \ref{lem:slope_curve} that for $\sum_{i=1}^{M}r_{ij}^{\text{inf}}>R$,  $\exists \:\: i$ such that the optimal rates $r_{ij}^{\text{opt}} < b_{ij}$. Thus, $p_{ss} \approx \frac{a_{ij}d_{ij} e^{\frac{a_{ij}b_{ij}}{2}}}{1-d_{ij}(1+e^{\frac{a_{ij}b_{ij}}{2}})} + \frac{a_{ij}e^{\frac{a_{ij}b_{ij}}{2}}}{(1+e^{\frac{a_{ij}b_{ij}}{2}})}$ is the optimal shadow price for the optimization problem in equation (\ref{eqn:opt_sub1}). Then, a small change in the shadow price $p(n)$ at the $n^{th}$ iteration can cause the rate $r_{ij}(n)$ (the root of $S_{ij}(r_{ij}) - p(n) =0$) to fluctuate between the concave and convex curvature of the slope curve $S_{ij}(r_{ij})$ for the $i^{th}$ UE. Therefore, it produces a fluctuation in the bid value $w_i(n)$ sent to the eNB, which in turn induces a vacillation of the shadow price $p(n)$ transmitted by eNB to the UEs. Hence, the iterative solution oscillates about the global optimal rates $r_{ij}^{\text{opt}}$.
\end{proof}

\begin{thm}\label{thm:sig_not_conv}
EURA Algorithms \ref{alg:UE_distributed} and \ref{alg:eNodeB_distributed} do not converge to the optimal solution for all eNB rates $R$.
\end{thm}
\begin{proof}
It directly follows from the corollaries \ref{cor:sig_convergence} and \ref{cor:sig_fluctuate} that the EURA algorithm does not converge to the global optimal solution for all values of $R$.
\end{proof}

The potential EURA seesawing about optimal rates and dearth of convergence thereof motivate us to include some robustness into the procedure. This is done in section \ref{sec:robust_alg}.

\subsection{EURA Robust Algorithm}\label{sec:robust_alg}
Incorporate robustness into the EURA Algorithms \ref{alg:UE_distributed} and \ref{alg:eNodeB_distributed} so that they converge for all eNB rates requires the algorithm to refrain from fluctuations in the non-convergent region for $\sum_{i=1}^{M}\underset{j}\max \:\:r_{ij}^{s} \ll R$. To do this, a \textit{fluctuation decay function} $\Delta w(n)$ as below reduces the step size between the current and previous bid, i.e. $w_i(n) - w_i(n-1)$, for every user $i$ if a fluctuation occurs. The allocated rates should coincide with the those of EURA Algorithms \ref{alg:UE_distributed} and \ref{alg:eNodeB_distributed} for $\sum_{i=1}^{M}\underset{j}\max \:\:r_{ij}^{s}>R$.
\begin{itemize}
\item \textit{Exponential function}: $\Delta w(n) = l_1 e^{-\frac{n}{l_2}}$.
\item \textit{Rational function}: $\Delta w(n) = \frac{l_3}{n}$.
\end{itemize}
where$l_1, l_2, l_3$ can be adjusted to change the bids $w_i$ decay rate.

\begin{rem}
The fluctuation decay function can be included in either UE EURA Algorithm or eNB EURA Algorithm.
\end{rem}

In our model, we choose to incorporate the decay function into the UE EURA Algorithm even though, as mentioned before, it can be placed in the eNB EURA as well. The fledgling robust EURA process is illustrated in Algorithms \ref{alg:UE_EURA} and \ref{alg:eNodeB_EURA}. Here, starting with $w_i(0) = 0$, each UE commences transmitting an initial bid $w_i(1)$ to the eNB, which at each iterate $n$ calculates the difference between the currently and formerly received bids $w_i(n)$ and $w_i(n-1)$, then exits if the difference falls below a threshold $\delta$; otherwise, it computes the shadow price $p_E(n) = \frac{\sum_{i=1}^{M}w_i(n)}{R}$ and send it to its covered UEs, amongst which The $i^{th}$ UE obtains the rate $r_{i}$ maximizing the statement $\log \beta_{i} V_i(r_{i}) - p_E(n)r_{i}$, estimates its new bid $w_i(n)=p_E(n) r_{i}(n)$, and sends it to the eNB.

\begin{algorithm}
\caption{UE Robust EURA Algorithm}\label{alg:UE_EURA}
\begin{algorithmic}
\STATE {Send initial bid $w_i(1)$ to eNB.}
\LOOP
	\STATE {Receive shadow price $p(n)$ from eNB.}
	\IF {STOP from eNodeB}
		\STATE {Calculate allocated rate $r_{ij} ^{\text{opt}}=\frac{w_i(n)}{p(n)}$.}
	\ELSE
		\STATE {Solve $r_{i}(n) = \arg \underset{r_i}\max \Big(\beta_i \log V_i(r_i) - p_E(n)r_i\Big)$.}
	        \STATE {Calculate new bid $w_i (n)= p(n) r_{i}(n)$.}
	        \IF {$|w_i(n) - w_i(n-1)| >\Delta w(n)$} %
			\STATE {$w_i(n) =w_i(n-1) + \text{sign}(w_i(n) -w_i(n-1))\Delta w(n)$}
	   	        \COMMENT {$\Delta w = l_1 e^{-\frac{n}{l_2}}$ or $\Delta w = \frac{l_3}{n}$}
	        \ENDIF
	        \STATE {Send new bid $w_i (n)$ to eNB.}
	\ENDIF
\ENDLOOP
\end{algorithmic}
\end{algorithm}

\begin{algorithm}
\caption{eNB EURA Algorithm}\label{alg:eNodeB_EURA}
\begin{algorithmic}
\LOOP
	\STATE {Receive bids $w_i(n)$ from UEs}
	\COMMENT{Let $w_i(0) = 1\:\:\forall i$}
	\IF {$|w_i(n) - w_i(n-1)|<\delta  \:\:\forall i$}
   		\STATE {STOP and allocate rates (i.e $r_{i}^{\text{opt}}$ to user $i$)}
	\ELSE
		\STATE {Calculate $p_E(n) = \frac{\sum_{i=1}^{M}w_i(n)}{R}$}
		\STATE {Send new shadow price $p_E(n)$ to all UEs}
	\ENDIF
\ENDLOOP
\end{algorithmic}
\end{algorithm}

\begin{rem}
If the subscriber differentiation parameter $\beta_{i}$ is available only at the eNB (or other network provider unit), the shadow price $p_E$ is changed to $\frac{p_E}{\beta_{i}}$.
\end{rem}

\subsection{IURA Algorithm}\label{sec:IURA_alg}
This section presents the second stage of the distributed resource allocation during which the application rates $r_{ij}$ are optimally assigned internally to the UEs in accordance with Algorithm \ref{alg:IURA}, where the $i^{th}$ UE leverages the EURA allocated rate $r_{i}^{\text{opt}}$ to solve $\textbf{r}_i =  \arg \underset{\textbf{r}_i}\max \sum_{j=1}^{N_i}(\alpha_{ij}\log U_{ij}(r_{ij})- p_I r_{ij}) +p_I r_{ij}^{\text{opt}}$.

\begin{algorithm}
\caption{UE IURA Optimization}\label{alg:IURA}
\begin{algorithmic}
\LOOP
      \STATE {Receive $r_i^{\text{opt}}$ from eNB.
      \COMMENT {by EURA Algorithms}}
      \STATE {Solve \\$\textbf{r}_i =  \arg \underset{\textbf{r}_i}\max \sum_{j=1}^{N_i}(\alpha_{ij}\log U_{ij}(r_{ij})- p_I r_{ij}) +p_I r_i^{\text{opt}}$}
      \COMMENT {$\textbf{r}_i = \{r_{i1}, r_{i2}, ..., r_{iN_i}\}$}
      \STATE {Allocate $r_{ij}$ to the $j^{th}$ application.}
\ENDLOOP
\end{algorithmic}
\end{algorithm}

Next, section \ref{sec:one_stage_alg} includes the centralized resource allocation algorithm in which the application rate are assigned in a monolithic stage.

\section{Centralized Algorithm}\label{sec:one_stage_alg}
The process for the centralized resource allocation (\ref{eqn:opt_multiapp}) consists of UE and eNB parts shown in respectively Algorithms \ref{alg:Centralized_UE} and \ref{alg:Centralized_eNodeB}, whose executions (Figure \ref{fig:multiple_app_flow_centralized}) start by UEs transmitting their application utility parameters to the eNB, which in turn solves the entire optimization by allotting the bandwidth to the applications in an optimum fashion. The rates, solutions to equation (\ref{sec:one_stage_global_optimal}), are the values $r_{ij}$ which solve the equation $\frac{\partial \log U_{ij}(r_{ij})}{\partial r_{ij}} = p(n)$ and are the intersection of the time varying shadow price, horizontal line $y = p(n)$, with the curve $y = \frac{\partial \log U_{ij}(r_{ij})}{\partial r_{ij}}$ geometrically.

\begin{algorithm}
\caption{UE Centralized Algorithm}\label{alg:Centralized_UE}
\begin{algorithmic}
\LOOP
      \STATE {Send application utility parameters $\{a_{ij}, b_{ij}, \alpha_{ij}, k_{ij}, r_{ij}^{\text{max}}\}$ to eNB.}
      \STATE {Receive rates $r_i^{\text{opt}} = \{r_{i1}^{\text{opt}}, r_{i2}^{\text{opt}},...,r_{iN_i}^{\text{opt}}\}$ from eNB.}
      \STATE {Allocate rate $r_{ij} ^{\text{opt}}$ internally to $j^{th}$ applications.}
\ENDLOOP
\end{algorithmic}
\end{algorithm}

\begin{algorithm}
\caption{eNB Centralized Algorithm}\label{alg:Centralized_eNodeB}
\begin{algorithmic}
\LOOP
      \STATE {Receive application utility parameters $\{a_{ij}, b_{ij}, \alpha_{ij}, k_{ij}, r_{ij}^{\text{max}}\}$  from UEs.}
      \STATE {Solve $\textbf{r} =  \arg \underset{\textbf{r}}\max \sum_{i=1}^{M}{\beta_{i}}\sum_{j=1}^{N_i}{\alpha_{ij}}\log U_{ij}(r_{ij}) - p(\sum_{i=1}^{M}\sum_{j=1}^{N_i}r_{ij} - R)$.}
      \COMMENT{where $\textbf{r} = \{r_{1}, r_{2}, ..., r_{M}\}$ and ${r}_i = \{r_{i1}, r_{i2}, ..., r_{iN_i}\}$}
      \STATE {Send ${r}_i = \{r_{i1}, r_{i2}, ..., r_{iN_i}\}$ to $i^{th}$ UE.}
\ENDLOOP
\end{algorithmic}
\end{algorithm}

\begin{figure}[t!]
\centering
  \includegraphics[width=\plotwidth]{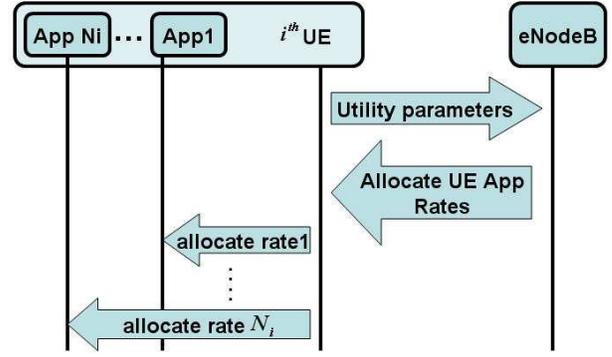}
  \caption{Centralized Algorithm: Resources are allocated to the applications running on the UEs in a monolithic stage, in which UEs transmit their application utility parameters to their eNB, which calculates the optimal application rates and transmit them to the germane UEs.}
  \label{fig:multiple_app_flow_centralized}
\end{figure}

Section \ref{sec:equivalence} proves that the distributed and centralized resource allocation methods are equivalent.

\section{Equivalence}\label{sec:equivalence}
Here, we show the mathematical equivalence of the distributed resource allocation in equations (\ref{eqn:opt_sub1}) and (\ref{eqn:opt_sub2}) with the centralized approach in equation (\ref{eqn:opt_multiapp}). First , lemma \ref{lem:opt_equivalent} is conceivable.

\begin{lem}\label{lem:opt_equivalent}
The aggregated and application utility slope curvature functions $S_{i}(r_{i}) = \frac{\partial \log V_{i}(r_{i})}{r_{i}}$ and $S_{ij}(r_{ij}) = \frac{\partial \log U_{ij}(r_{ij})}{r_{ij}}$ are invertible and their inverse functions $r_{i} = S_{i}^{-1}(.)$ and $r_{ij} = S_{ij}^{-1}(.)$ are strictly decreasing.
\end{lem}

\begin{proof}
The concavity of the logarithmic utility $U_ij$ yields in $U'_{ij}(r_{ij}) = \frac{\partial U_{ij}(r_{ij})}{\partial r_{ij}} > 0$ and $U''_{ij}(r_{ij}) = \frac{\partial ^2U_{ij}(r_{ij})}{\partial r_{ij}^2} < 0$, and lemma \ref{lem:concavity} stems out $S_{ij}(r_{ij}) = \frac{\partial \log(U_{ij}(r_{ij}))}{\partial r_{ij}} =  \frac{U'_{ij}(r_{ij})}{U_{ij}(r_{ij})} > 0$ and $\frac{\partial S_{ij}(r_{ij})}{\partial r_{ij}} =  \frac{U''_{ij}(r_{ij})U_{ij}(r_{ij})-U'^2_{ij}(r_{ij})}{U^2_{ij}(r_{ij})} < 0$. Also, for the utility function, we have $ U_{ij}(r_{ij}) > 0$, $U_{ij}(r_{ij})$ is increasing, and it is twice differentiable with respect to $r_{ij}$ (section \ref{sec:Problem_formulation}). Therefore, $S_{ij}(r_{ij})$ of the logarithmic utility function is strictly decreasing. From equation (\ref{eqn:sigmoid_derivative}), for the sigmoidal utility function $U_{ij}(r_{ij})$ where  $0 < r_{ij} < R$, we can write inequality set (\ref{eqn:sigmoid_derivative}), giving that $S_{ij}(r_{ij})$ of the sigmoidal utility 
function is strictly decreasing.

\begin{equation}\label{eqn:sigmoid_derivative}
\begin{aligned}
S_{ij}(r_{ij}) > 0, \frac{\partial}{\partial r_{ij}}S_{ij}(r_{ij}) < 0
\end{aligned}
\end{equation}

Equation (\ref{eqn:sum_slope}) and inequalities \ref{eqn:sigmoid_derivative} yield in inequalities (\ref{eqn:sum_slope_curvature_ineq1}). Henceforth, $S_{ij}(r_{ij})$ and $S_{i}(r_{i})$ of all the utilities in section \ref{sec:Problem_formulation} are strictly decreasing functions; thereby, the slope curvature functions $S_{ij}(r_{ij})$ and $S_{i}(r_{i})$ are invertible and the inverse functions are strictly decreasing.

\begin{equation}\label{eqn:sum_slope_curvature_ineq1}
\begin{aligned}
S_{i}(r_{i}) = \sum_{j=1}^{N^{S}_i}\alpha_{ij}S_{ij}(r_{ij})+\sum_{j=N^{S}_i+1}^{N_i}\alpha_{ij}S_{ij}(r_{ij}) >0\\
\frac{\partial S_{i}(r_{i})}{\partial r_{i}} = \sum_{j=1}^{N^{S}_i}\alpha_{ij}\frac{\partial S_{ij}(r_{ij})}{\partial r_{ij}}+\sum_{j=N^{S}_i+1}^{N_i}\alpha_{ij}\frac{\partial S_{ij}(r_{ij})}{\partial r_{ij}} <0
\end{aligned}
\end{equation}
\end{proof}

\begin{cor}\label{cor:opt_equivalence}
The optimal rates assigned by the distributed optimization in equations (\ref{eqn:opt_sub1}) and (\ref{eqn:opt_sub2}) is equal to the one allocated by the centralized optimization in equation (\ref{eqn:opt_multiapp}).
\end{cor}

\begin{proof}
The centralized optimization's (equation(\ref{eqn:opt_multiapp_log})) lagrangian can be written as equation (\ref{eqn:opt_multiapp_log_Lag}) where $z \geq 0$ is the slack variable and $p_T$ is the lagrange multiplier.

\begin{equation}\label{eqn:opt_multiapp_log_Lag}
L_T (\textbf{r}) =  (\sum_{i=1}^{M}{\beta_{i}}\sum_{j=1}^{N_i}{\alpha_{ij}}\log U_{ij}(r_{ij})) - p_T ( \sum_{i=1}^{M}\sum_{j=1}^{N_i}r_{ij} - R + z)
\end{equation}

Then, we have that:

\begin{equation}\label{eqn:opt_multiapp_log_Lag_diff}
\frac{\partial L_T(\textbf{r})}{\partial r_{ij}}  =  \beta_{i} \alpha_{ij} S_{ij}(r_{ij}) - p_T = 0 \Rightarrow p_T =  \beta_{i} \alpha_{ij} S_{ij}(r_{ij})
\end{equation}

so, the $i^{th}$ UE's $j^{th}$ application rate is:

\begin{equation}\label{eqn:opt_multiapp_log_inv}
r_{ij} = S_{ij}^{-1}(\frac{p_T}{\beta_{i} \alpha_{ij}})
\end{equation}

Using equation (\ref{eqn:sum_slope}), we can write:

\begin{equation}\label{eqn:opt_multiapp_log_p_with_N_i}
N_ip_T =  \beta_{i} S_{i}(r_{i})
\end{equation}

And the $i^{th}$ UE rate can be calculated as equation (\ref{eqn:opt_multiapp_log_inv_rate}).

\begin{equation}\label{eqn:opt_multiapp_log_inv_rate}
r_{i} = S_{i}^{-1}(\frac{N_ip_T}{\beta_{i}}).
\end{equation}

The EURA optimization's (equation (\ref{eqn:opt_sub1_log})) lagrangian can be written as equation (\ref{eqn:opt_sub1Lag}) where $z \geq 0$ is the slack variable and $p_E$ is the lagrange multiplier.

\begin{equation}\label{eqn:opt_sub1Lag}
L_E (\textbf{r}) =  (\sum_{i=1}^{M}{\beta_{i}}\log V_{i}(r_{i})) - p_{E} (\sum_{i=1}^{M}r_{i} - R + z)
\end{equation}

Then, we have that:

\begin{equation}\label{eqn:opt_sub1_log_Lag_diff}
\frac{\partial L_E(\textbf{r})}{\partial r_{i}}  =  \beta_{i}S_{i}(r_{i}) - p_{E} = 0 \Rightarrow p_E =  \beta_{i} S_{i}(r_{i})
\end{equation}

So, the $i^{th}$ UE rate is:

\begin{equation}\label{eqn:opt_multiapp_log_inv}
r_{i} = S_{i}^{-1}(\frac{p_E}{\beta_{i}})
\end{equation}

Replacing $S_i$ from equation (\ref{eqn:sum_slope}), we can write:

\begin{equation}\label{eqn:opt_sub1_log_p_f}
p_E =  \beta_{i}\sum_{j=1}^{N_i}\alpha_{ij} S_{ij}(r_{ij})
\end{equation}

And, we get equation (\ref{eqn:opt_sub1_1_log_p_f}) below.

\begin{equation}\label{eqn:opt_sub1_1_log_p_f}
p_E = \sum_{j=1}^{N_i} p_T=  {N_i}p_T
\end{equation}

Equations (\ref{eqn:opt_multiapp_log_inv_rate}) and (\ref{eqn:opt_multiapp_log_inv}) signify that the centralized and EURA optimizations lead to identical UE rates.

The IURA optimization's (equation \ref{eqn:opt_sub2_log}) lagrangian can be written as equation (\ref{eqn:opt_sub2LagIURA}) where $z \geq 0$ is the slack variable and $p_I$ is the lagrange multiplier corresponding to the internal shadow price, price per bandwidth for all applications in the $i^{th}$ UE.

\begin{equation}\label{eqn:opt_sub2LagIURA}
L_I (r_{i}) =  (\sum_{j=1}^{N_i}{\alpha_{ij}}\log U_{ij}(r_{ij})) - p_I (\sum_{j=1}^{N_i}r_{ij} - r_{ij}^{\text{opt}} + z)
\end{equation}

Then, we have that:

\begin{equation}\label{eqn:opt_sub2_log_Lag_diff}
\frac{\partial L_I(r_{i})}{\partial r_{ij}}  =  \alpha_{ij} S_{ij}(r_{ij}) - p_I = 0 \Rightarrow p_I =  \alpha_{ij}  S_{ij}(r_{ij}) \:\:\:\forall \:\:j
\end{equation}

And, summing the $i^{th}$ UE applications gives that:

\begin{equation}\label{eqn:opt_sub2_log_p_f}
\sum_{j=1}^{N_i}p_I =  \sum_{j=1}^{N_i}\alpha_{ij}  S_{ij}(r_{ij})
\end{equation}

Using equation (\ref{eqn:sum_slope}) results in equation (\ref{eqn:opt_sub2_log_p_f}).

\begin{equation}\label{eqn:opt_sub2_log_p_f}
\beta_{i}{N_i}p_I = \beta_{i}S_{i}(r_{i}) = p_E = {N_i}p_T \Rightarrow p_T = \beta_{i} p_I
\end{equation}

So, $i^{th}$ UE's $j^{th}$ application rate can be written as equation (\ref{eqn:opt_sub2_log_p_f}).

\begin{equation}\label{eqn:abc}
r_{ij} = S^{-1}_{ij}(\frac{p_I}{\alpha_{ij}}) = S^{-1}_{ij}(\frac{p_T}{\beta_{i}\alpha_{ij}})
\end{equation}

Considering the constraints of the equation (\ref{eqn:opt_sub2_log}), the total rate of the $i^{th}$ UE can be written as equation (\ref{eqn:opt_sub2_log_inv_rate}).
\begin{equation}\label{eqn:opt_sub2_log_inv_rate}
r^{\text{opt}}_{i} = \sum_{j=1}^{N_i} r_{ij} = \sum_{j=1}^{N_i}S_{ij}^{-1}(\frac{p_T}{\beta_{i}\alpha_{ij}})
\end{equation}

Equations (\ref{eqn:abc}) and (\ref{eqn:opt_sub2_log_p_f}) signify that the centralized and IURA optimizations lead to identical application rates. As such, the UE and application rates assigned by the centralized and distributed optimizations are the same.
\end{proof}

\begin{thm}\label{thm:multiapp_dist}
The distributed optimization in equations (\ref{eqn:opt_sub1}) and (\ref{eqn:opt_sub2}) is equivalent to the centralized optimization in equation (\ref{eqn:opt_multiapp}).
\end{thm}
\begin{proof}
Stemming out equal rates (Corollary \ref{cor:opt_equivalence}) indicates the distributed and centralized optimizations are equivalent.
\end{proof}

\section{Simulation Results}\label{sec:sim}
A cell with $M = 6$ UEs and an eNB, depicted in Figure \ref{fig:system_model}, is considered and each UE concurrently runs a delay-tolerant and a real-time application with respectively logarithmic and sigmoidal utility functions with parameters in Table \ref{table:parameters}. The sigmoidal utility with parameters $a = 5$, $b=10$ approximates a step function at rate $r =5$ and is a good model for Voice-over-IP (VoIP), while parameters $a = 3$, $b=15$ is an approximation of a real-time application with an inflection point at rate $r=15$ and is conducive to modeling standard definition video streaming, whereas parameters $a = 1$,  $b=25$ is an estimation of another real-time application with the inflection point $r=25$ and is appropriate for the high definition video streaming. Moreover, the logarithmic utilities with $r^{\text{max}} = 100$ and distinct $k_i$ parameters estimate delay-tolerant FTP applications. The plots of the utility functions in Table \ref{table:parameters} are shown in Figure \ref{fig:sim:
app_utilities}, from which we can observe that the real-time applications require a minimum rate, i.e. the inflection point, after which the application QoS is fulfilled to a large extent. On the other hand, the logarithmic utility is provided with some QoS even at low rates suitable for the delay-tolerant nature of the applications. Furthermore, as we can observe from Figure \ref{fig:sim:app_utilities}, in compliance with the properties mentioned in section \ref{sec:Problem_formulation} the utility functions are strictly increasing continuous functions, zero valued at zero rates. Furthermore, the first derivative of the utility functions natural logarithm, $S_ij(r_ij)$, are shown in Figure \ref{fig:sim:diff_log_app_utilities}, which reflects the positivity and decreasing nature of the first derivative in line with lemmas \ref{lem:opt_equivalent} and \ref{lem:opt_equivalent}.

\begin{figure}
\centering
\subfigure[Application Utility Functions]{\label{fig:sim:app_utilities}\includegraphics[width=\plotwidth]{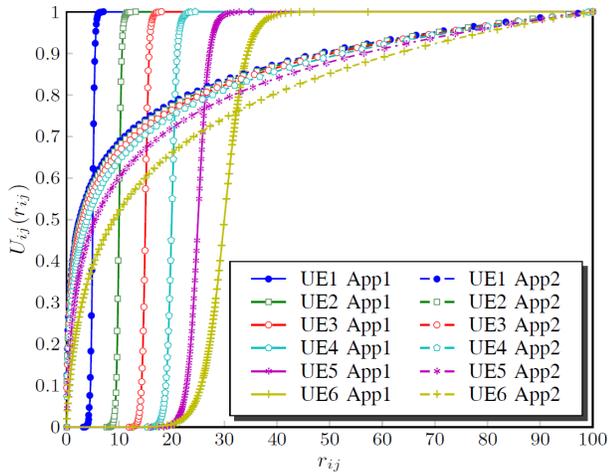}}\qquad
\subfigure[Utility Slope Curvature Functions]{\label{fig:sim:diff_log_app_utilities}\includegraphics[width=\plotwidth]{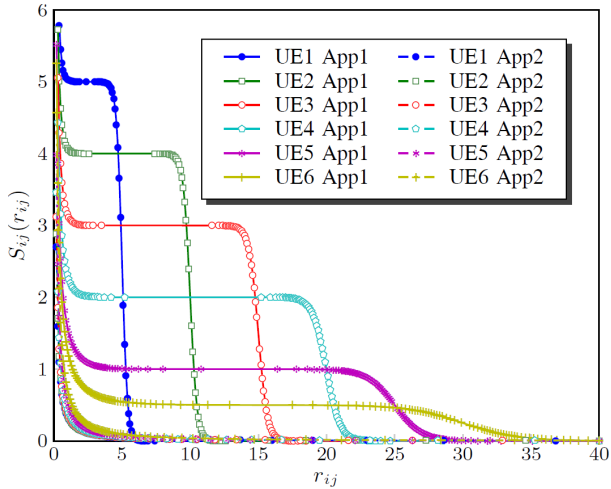}}%
\caption{The system contains $6$ UEs, each concurrently running a delay-tolerant and real-time application with respective identically colored logarithmic and sigmoidal utility functions $U_{ij}$ vs. the application-assigned rates $r_{ij}$ plots in Figure \ref{fig:sim:app_utilities}. Utility slope curvature functions, the first derivative of the application utility natural logarithms $S_{ij}$ with respect to the application rates $r_{ij}$ are illustrated in Figure \ref{fig:sim:diff_log_app_utilities} where identical colors relate to the applications on one UE.}
\end{figure}

Then, the distributed resource allocation approach (Algorithms \ref{alg:UE_EURA}, \ref{alg:eNodeB_EURA}, and \ref{alg:IURA}) and the centralized rate assignment procedure (Algorithms \ref{alg:Centralized_UE} and \ref{alg:Centralized_eNodeB}) were applied to the aforesaid logarithmic and sigmoidal utility functions using MATLAB. To account for the applications usage percentage, we set the application status weight vector in equation (\ref{eqn:utility_agg}) as $\boldsymbol\alpha = \{\alpha_{11}, \alpha_{21}, \alpha_{31}, \alpha_{41}, \alpha_{51}, \alpha_{61},\alpha_{12}, \alpha_{22}, \alpha_{32}, \alpha_{42}, \alpha_{52}, \alpha_{62}\}$ where $\alpha_{ij}$ represents the status weight of the $j^{th}$ application of the $i^{th}$ UE. It is noteworthy that the addition of application usage percentages per UE is unity, i.e. $\alpha_{i1}+\alpha_{i2}=1$.

\begin{figure}
\centering
\subfigure[Aggregated Utility Functions]{\label{fig:sim:users_utilities}\includegraphics[width=\plotwidth]{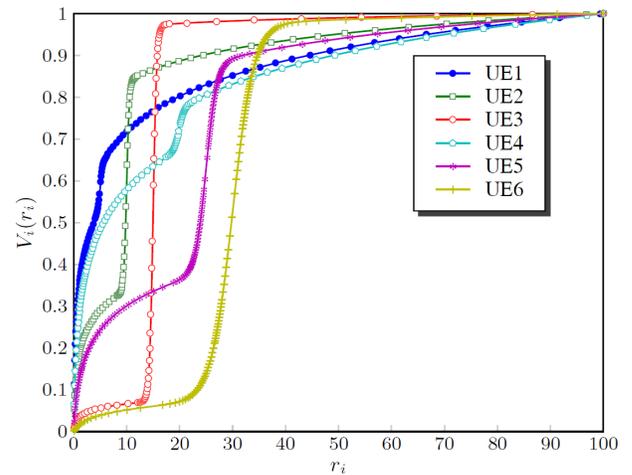}}\qquad
\subfigure[Aggregated Slope Curvature Functions]{\label{fig:sim:diff_log_users_utilities}\includegraphics[width=\plotwidth]{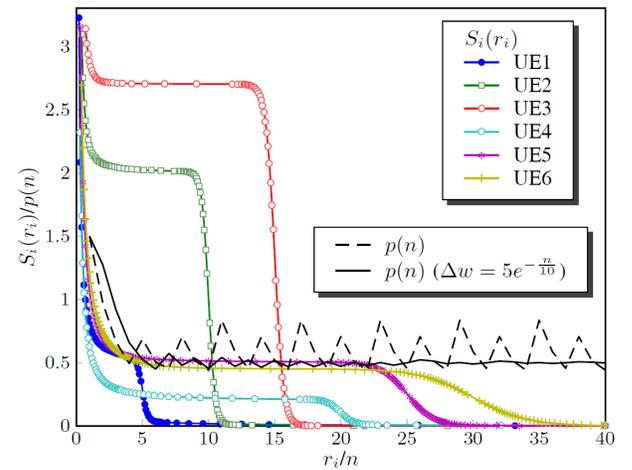}}%
\caption{Figure \ref{fig:sim:users_utilities} plots the aggregated utilities, multiplications of the usage-percentage-powered application utility functions $V_i(r_i)$ vs. the UE rates $r_i$, where $i \in \{1,...,6\}$. Figure \ref{fig:sim:diff_log_users_utilities} illustrates the aggregated slope curvature functions , first derivative of the the aggregated utility natural logarithms $S_i(r_i)$. Furthermore, decay function-induced robustness effect is depicted; As we can see, the lack of decay functions yields in the system instability revealed in the shadow price oscillation.}
\end{figure}

In addition, the aggregated utility functions $V_i(r_i)$ for $i \in \{1,...,6\}$ are depicted in Figure \ref{fig:sim:users_utilities} and the first derivative of their natural logarithm, $S_i(r_i)$ for $i \in \{1,...,6\}$, are illustrated in Figure \ref{fig:sim:diff_log_users_utilities}. As we can see, in compliance with lemma \ref{lem:slope_curve}, the slope curvature functions inflection points occur at the application utility functions' inflection points. Furthermore, in line with lemma \ref{lem:opt_equivalent}, the slope curvature functions are strictly decreasing.

\begin {table}[]
\caption {Applications Utility Parameters}
\label{table:parameters}
\begin{center}
\renewcommand{\arraystretch}{1.4}
\begin{tabular}{| l || l |}
  \hline
  \multicolumn{2}{|c|}{Applications Utilities Parameters} \\  \hline
  UE1 App1 & Sigmoid $a=5,\:\: b=5$  \\ \hline
  UE2 App1 & Sigmoid $a=4,\:\: b=10$ \\ \hline
  UE3 App1 & Sigmoid $a=3,\:\: b=15$ \\ \hline
  UE4 App1 & Sigmoid $a=2,\:\: b=20$ \\ \hline
  UE5 App1 & Sigmoid $a=1,\:\: b=25$\\ \hline
  UE6 App1 & Sigmoid $a=0.5,\:\: b=30$ \\ \hline
  UE1 App2 & Logarithmic $k=15,\:\: r^{\max}=100$  \\ \hline
  UE2 App2 & Logarithmic $k=12,\:\: r^{\max}=100$  \\ \hline
  UE3 App2 & Logarithmic $k=9,\:\: r^{\max}=100$   \\ \hline
  UE4 App2 & Logarithmic $k=6,\:\: r^{\max}=100$   \\ \hline
  UE5 App2 & Logarithmic $k=3,\:\: r^{\max}=100$   \\ \hline
  UE6 App2 & Logarithmic $k=1,\:\: r^{\max}=100$ \\ \hline
\end{tabular}
\end{center}
\end {table}

Next section investigates bids and rate allocations for the UEs and applications in our system under varying eNB resource availabilities.

\subsection{Rate Allocation and Bids for $10\le R\le200$}
In the following simulations, we set the termination threshold $\delta = 10^{-4}$ and the eNB rate $R$ to sweep from 10 to 200 with an step size of 5 bandwidth units. Besides, the application status weights is considered to be $\boldsymbol\alpha = \{0.1, 0.5, 0.9, 0.1, 0.5, 0.9, 0.9, 0.5, 0.1, 0.9, 0.5, 0.1\}$. It is worth mentioning that the addition of usage percentages per UE is unity, e.g. adding the $1st$ and the $6th$ components of the set (which are indeed the usage percentages for both of applications running on UE1), we get $0.1+0.9=1$, and so forth. For the distributed resource allocation (Algorithms \ref{alg:UE_EURA}, \ref{alg:eNodeB_EURA}, and \ref{alg:IURA}), UE assigned rates and pledged bids are depicted in Figure \ref{fig:sim:user_allocated_rates} during the EURA Algorithm with the changes in the eNB available resources $R$. As we can observe, initially all the UEs are allocated some rates which is owing to the fact that they all subsume real-time applications in need of immediate rate 
allocations before any QoS is met. For instance, UE2 has a real-time streaming video application (based on Table \ref{table:parameters}), which requires a bandwidth assignment right away. In Figure \ref{fig:sim:user_allocated_bids}, we show the UEs bids $\{w_i | i \in \{1, ... , 6\}\}$ during the EURA algorithm under changing eNB bandwidth $R$. First of all, we see that the more resources become available at the eNB, the higher rates are assigned to the UEs. On the other hand, the dearth of the resources (small $R$) causes those UEs which have applications with higher bit rate requirements to bid higher in order to gain resources. For instance, since UE2 includes a real-time streaming video application, its urgent need for bandwidth allocation causes its initial higher bid for the resources, which is responded by its fast allocation portrayed in the Figure \ref{fig:sim:user_allocated_rates}.

\begin{figure}
\centering
\subfigure[UE Optimal Rates]{\label{fig:sim:user_allocated_rates}\includegraphics[width=\plotwidth]{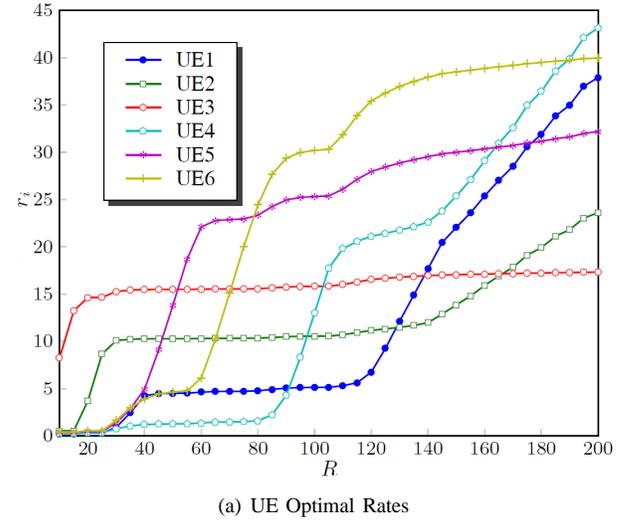}}\qquad
\subfigure[UE Bids]{\label{fig:sim:user_allocated_bids}\includegraphics[width=\plotwidth]{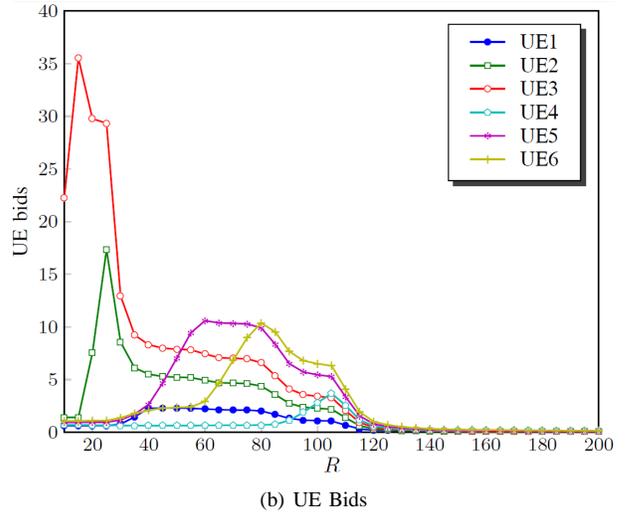}}%
\caption{Figure \ref{fig:sim:user_allocated_rates} depicts the optimal rates allocated to the UEs by the distributed scheme vs. eNB resources. No user is dropped as no assignment is zero. Figure \ref{fig:sim:user_allocated_bids} illustrates the UE bids for acquiring the resources vs. the eNB rate. The applications requiring more resources bid higher. When bandwidth is scarce, applications needing more resources bid significantly higher than the others. The plots reveal that the higher bid are tantamount to receiving more resources.}
\end{figure}

Then, the IURA algorithm has the UEs internally allocate rates to their applications based on the pledged bids as illustrated in Figures \ref{fig:sim:app_allocated_rates} and \ref{fig:sim:app_allocated_bids}. In Figure \ref{fig:sim:app_allocated_rates}, we show the allocated applications rates $\{r_{ij} | i \in \{1, ... , 6\} \wedge j \in \{1,2\}\}$ during the IURA algorithm under changing eNB rate $R$. As we can observe, initially more resources are allocated to the real-time applications since these have more stringent QoS requirements. In Figure \ref{fig:sim:app_allocated_bids}, we illustrate the applications' internally pledged bids $\{w_{ij} | i \in \{1, ... , 6\} \wedge j \in \{1,2\}\}$ during the IURA algorithm under changing eNB rate $R$. Inasmuch as the real-time applications of the UEs need more resources, they bid higher than the delay-tolerant applications specially when the resources are scarce. In fact, we can see that the bid values for the delay-tolerant applications is significantly less 
than those of the real-time ones such that they are very close to the horizontal axis in Figure \ref{fig:sim:app_allocated_bids}. Furthermore, those applications with higher QoS requirements such as the real-time streaming video in UE1 (red plot) bid higher in order to gain more bandwidth. However, as more resources become available at the eNB, bid values slash down as well.

\begin{figure}
\centering
\subfigure[Application Optimal Rates]{\label{fig:sim:app_allocated_rates}\includegraphics[width=\plotwidth]{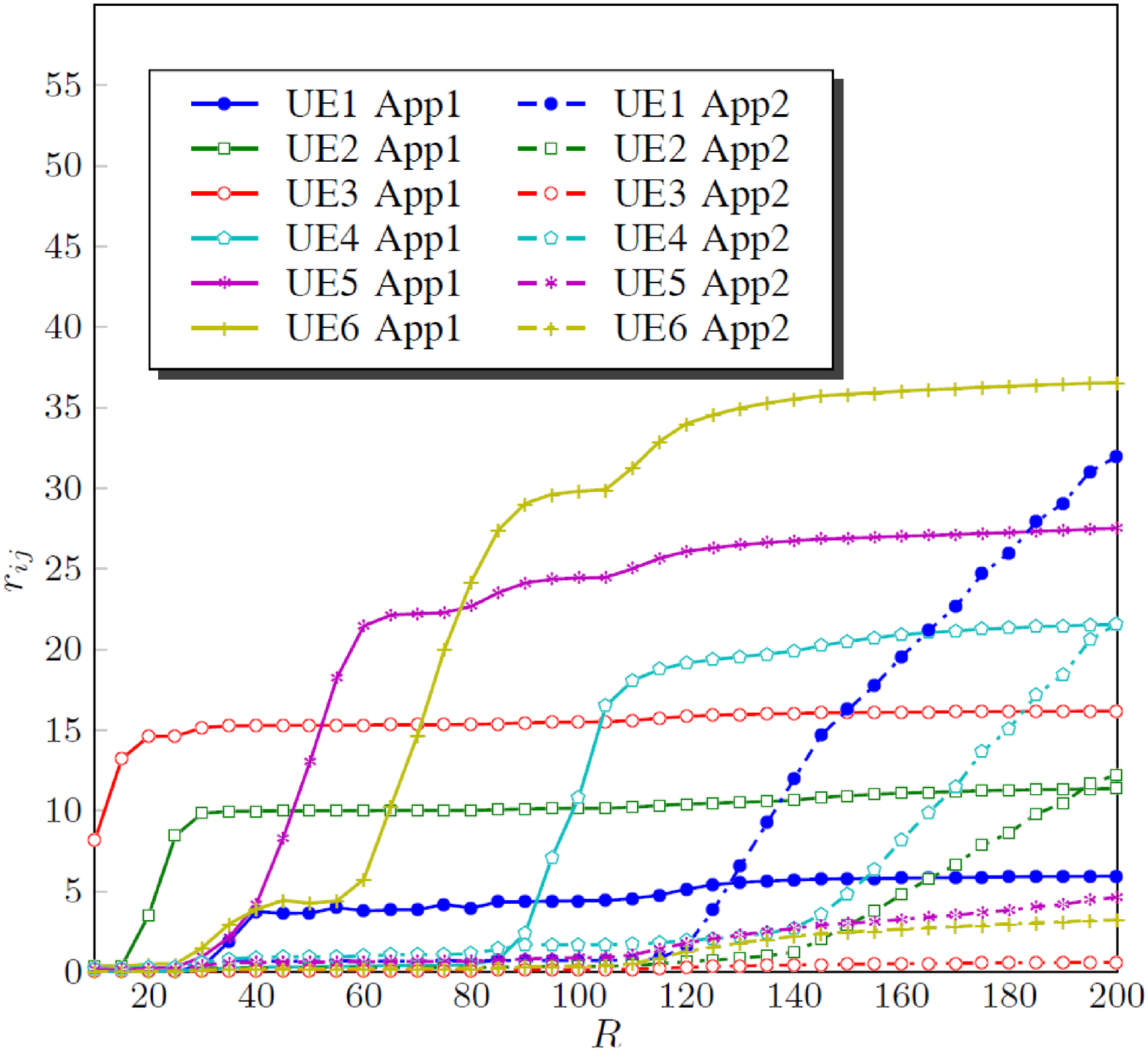}}\qquad
\subfigure[Application Bids]{\label{fig:sim:app_allocated_bids}\includegraphics[width=\plotwidth]{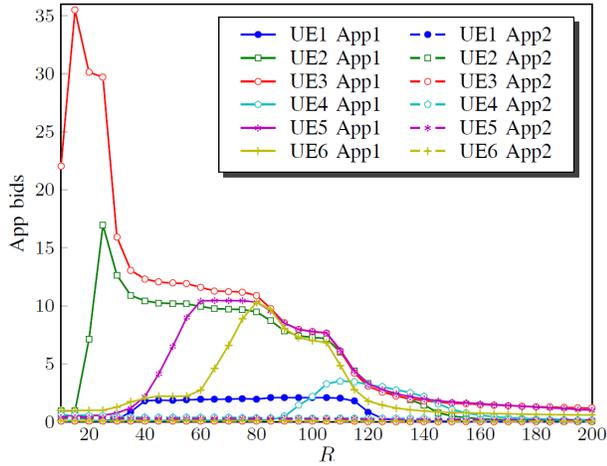}}%
\caption{Figure \ref{fig:sim:app_allocated_rates} depicts the optimal application rates $r_{ij}$ vs. the eNB rate $R$. Applications running on an UE are identically colored. As we can see, real-time applications are initially allocated more resources as opposed to the delay-tolerant ones due to their urgent need for resources. Figure \ref{fig:sim:app_allocated_bids} illustrates the applications bids in the UEs. The real-time applications bid higher when the resources are scarce, while the opulence of eNB resources escalates the application rates and reduces the UE bids.}
\end{figure}

On the contrary, the centralized resource allocation (Algorithms \ref{alg:Centralized_UE} and \ref{alg:Centralized_eNodeB}) assigns the application rates directly by the eNB, and the rates and bids are equal to the ones allocated by the distributed allocation in conformance to the theorem \ref{thm:multiapp_dist}. Running the simulations, we got the same rate and bid diagrams as in the in Figures \ref{fig:sim:app_allocated_rates} and \ref{fig:sim:app_allocated_bids}. It is notable that since utility proportional fairness objective functions are leveraged in the formation of the optimizations in equations (\ref{eqn:opt_multiapp}), (\ref{eqn:opt_sub1}), and (\ref{eqn:opt_sub2}), both the distributed and the centralized algorithms do not assign a zero rate to any UEs, thereby no user is dropped and a minimum QoS is warranted. As we mentioned before, an eNB allocates the majority of the resources to the real-time applications until they reach their utility inflection rate $r_{ij} = b_{ij}$. However, when the 
total eNB rate exceeds the inflection point rates sum $\sum b_{ij}$ of all real-time applications incumbent in the system, eNB can allot more resources to the delay-tolerant applications with ease of mind. This behavior is observed with the rate increase and bid value plummet that take place after the eNB rate surpasses the inflection points sum, i.e. $R = \sum b_{ij} = 105$, in Figure \ref{fig:sim:app_allocated_rates}.

Furthermore, the improvement in the Algorithms \ref{alg:UE_EURA} and \ref{alg:eNodeB_EURA} over the Algorithms \ref{alg:UE_distributed} and \ref{alg:eNodeB_distributed} can be observed in the fluctuation reduction of the shadow price depicted in Figure \ref{fig:sim:diff_log_users_utilities}, in which the decay function stabilizes the rate allocation by eliminating oscillations. Such an allocation behavior is similarly seen for Algorithm \ref{alg:UE_EURA} and \ref{alg:eNodeB_EURA} over Algorithm \ref{alg:UE_distributed} and \ref{alg:eNodeB_distributed} for $R > \sum b_{ij} = 105$, but Algorithm \ref{alg:UE_EURA} and \ref{alg:eNodeB_EURA} fails to assign the optimal rates and bids for $R < \sum b_{ij} = 105$. Therefore, Algorithm \ref{alg:UE_EURA} and \ref{alg:eNodeB_EURA} is robust under scarce resource availability circumstances.

Next, section \ref{sec:Pricing} discusses the pricing capability of the proposed resource allocation modi operandi, and presents germane simulation results.

\subsection{Pricing for $10\le R\le200$} \label{sec:Pricing}
As we explained before, Figure \ref{fig:sim:app_allocated_rates} shows the final rates and bids of different applications with varying eNB bandwidth, and the applications bids are proportional to the allocated rates. For example, the real-time applications (sigmoidal utilities) bid higher when the eNB resources are scarce and their bids reduce as $R$ increases. Therefore, the pricing, proportional to the bids, is \textit{traffic-dependent} which outfits service providers with the option to escalate the service price for their subscribers when the traffic load on the system is high. Thereby, service providers can motivate mobile subscribers to utilize the network when the traffic load is low in that they will be paying less for the same services by using the network during off-peak hours.

The shadow price $p(n)$, representing the total price per unit bandwidth for all users and applications, is illustrated in Figure \ref{fig:multiple_app_shadow_price} when eNB rate changes. As we can observe, the price is high under high-traffic situations, implied by a fixed number of users with less available resources ($R$ is small), and it decreases for low-traffic circumstances when the same number of users have the luxury of more resources ($R$ is large). It is particularly noticeable that large plummets in the shadow price occur after $R=\{15, 25, 85, 105\}$ which are essentially the points at which the rate for one of the real-time application utilities exceeds that of its inflection point. Furthermore, a large decrease is visible at the sum of the inflection points, i.e. $\sum_{i=1}^{k} r_{ij}^{\text{inf}}$. Here, $k = \{1,2,...,M\}$ is the users index, $M$ is the number of users, and $i$ is the user with the maximum utility slope $\arg\max_i S_i(r_i)$, in our case user 3 ($b_{3j}=15$) followed by 
user 2 ($b_{2j}=10$) then the three users 1, 5, 6 which have almost the same $S_i(r_i)$ ($b_{1j}=5, b_{5j}=25, b_{6j}=30$), and ultimately user 4 ($b_{4j}=20$). The larger the difference between slopes $\Delta S_{ij} =  |S_i(r_i)-S_j(r_j)|$, the higher the change in the shadow price $p(n)$ plot vs. $R$.

\begin{figure}[t!]
\centering
  \includegraphics[width=\plotwidth]{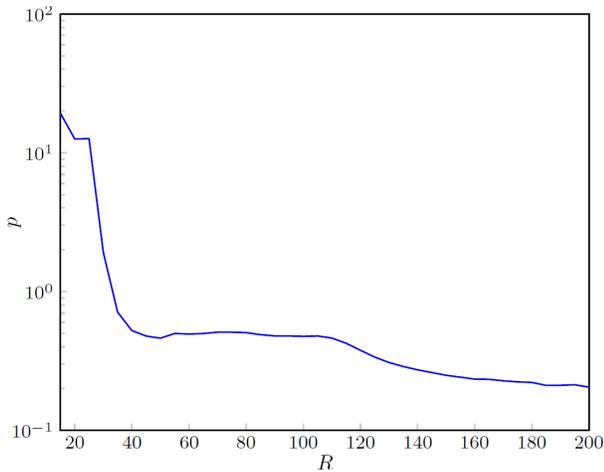}
  \caption{Shadow Price $p$ vs. eNB Resources $R$: Availability of more eNB resources reduces the shadow price.}
  \label{fig:multiple_app_shadow_price}
\end{figure}
\section{Conclusion}\label{sec:concl}
In this paper, we introduced a novel QoS-minded centralized and a distributed algorithm for the resource allocation within the cells of cellular communications systems. We formulated the centralized and distributed approaches as respectively a singular and double utility proportional fairness optimization problems, where the former allocated running applications rates directly by the allocation entity such as an eNB in response to the UE utility parameters sent, whereas the latter assigned the UE rates by the eNB in its first stage followed by the application rate allocation by the UEs in its second stage. Users ran both delay-tolerant and real-time applications mathematically modelled correspondingly as logarithmic and sigmoidal utility functions, where the function values represented the applications QoS percentage. Both of the proposed resource allocation formulations incorporated the service differentiation, application status differentiation modelling the applications usage percentage, and subscriber 
differentiations amongst subscribers priority within networks into their formulation. Not only did we prove that the proposed resource allocation problems were convex and solved them through Lagrangian of their dual problems, but also we proved the optimality of the rate assignments and the mathematical equivalence of the proposed distributed and centralized resource allocation schemes. Furthermore, we proved the mathematical equivalence of the distributed and centralized approaches by showing the both methods yield in identical optimal rates and pledged bids during their resource allocation processes.

Furthermore, we analyzed the algorithm convergence under varying sums of resources available to the eNB and introduced robustness into the distributed algorithm by incorporating decay functions into the aforementioned algorithm so that it converged to optimal rates for both high and low traffic loads occurring during the day by damping the rate assignment fluctuations resulting from scarcity of resources that appeared particularly in peak-traffic circumstances.

\bibliographystyle{ieeetr}
\bibliography{pubs}
\end{document}